%% file: quantum_oracle.tex
\documentclass[11pt]{article}
\usepackage[margin=1in]{geometry}
\input{setup.tex}

\usepackage{amssymb}
\usepackage{tikz}
\usetikzlibrary{calc}
\usetikzlibrary{positioning}
\usepackage{hyperref}
\usepackage{xurl}

\hypersetup{
colorlinks = true,
citecolor= blue,
linkcolor= blue,
breaklinks=true
}
\title{Non-Standard Oracles for Bounded-Error Complexity Classes}
\date{}

\author{Avantika Agarwal\\
    \small Institute for Quantum Computing\\
    \small University of Waterloo\\
    \small \texttt{avantika.agarwal@uwaterloo.ca}
    \and Srijita Kundu\\
    \small Quantum Computing Research Centre\\
    \small Hon Hai (Foxconn) Research Institute\\
    \small \texttt{srijita.kundu@foxconn.com.sg}}

\begin{document}

\maketitle

\begin{abstract}
In recent years, the quantum oracle model introduced by Aaronson and Kuperberg (2007) has found a lot of use in showing oracle separations between complexity classes and cryptographic primitives. It is generally assumed that proof techniques that do not relativize with respect to quantum oracles will also not relativize with respect to classical oracles. Aaronson (2009) showed that this is not the case by showing a complexity class containment that relativizes with respect to classical oracles but not with quantum oracles. However, their result only works for zero-error quantum complexity classes and they leave open the problem for bounded-error complexity classes. We show that there is a quantum oracle problem that is contained in the class $\QMA$, but not in a class we call $\polyQCPH$. However, with respect to classical oracles, $\QMA$ is contained in $\polyQCPH$, because $\polyQCPH$ is equal to $\PSPACE$ with respect to classical oracles. Our result works for $\polyQCPH$, which is a bounded-error complexity class, thus it resolves the open problem from Aaronson (2009).

We also show that the same separation holds relative to a distributional oracle, which is a model introduced by Natarajan and Nirkhe (2024). We believe our findings show the need for some caution when using these non-standard oracle models, particularly when showing separations between quantum and classical resources.
\end{abstract}

\section{Introduction}
In complexity theory, we try to compare the power of different computational resources. For example, we want to understand whether quantum resources are any more powerful than classical resources. Often it is beyond the reach of current techniques to show such results unconditionally, so we try to show them in the oracle model. Oracle separations in the classical oracle model have been very well-studied for many interesting complexity classes \cite{BGS75, Has86, HRST17, BBBV97, RT19}. One of the motivations for showing oracle separations between complexity classes is the following: if two classes can be separated in the relativized setting (with respect to oracles), then a relativizing proof technique (a technique insensitive to the presence of oracles) certainly cannot prove they are equal. This is known as the ``relativization barrier''. In our work, we discuss the problem of relativization with respect to quantum and distributional oracles, specifically with the lens of bounded-error quantum complexity classes.

With the goal of constructing an oracle separation between the complexity classes $\QMA$ and $\QCMA$, \cite{AK07} introduce the model of quantum oracles and show a quantum oracle separation between $\QMA$ and $\QCMA$. Following their work, there have been multiple results in quantum complexity theory and quantum cryptography which show oracle separations with respect to quantum oracles. For example, \cite{Aar09} show a quantum oracle separation between $\QMA$ and $\QMAone$. Separations between $\QMA$ and $\QCMA$ with slightly different quantum oracles are also shown in \cite{FK18, BFM23}. Recent works have tried to study whether $\QMA$ satisfies properties similar to $\NP$, and have shown that this is not true for many natural properties of $\NP$ in the quantum oracle model. For example, \cite{INNRY22} show that relative to a quantum oracle, $\QMA$-search does not reduce to $\QMA$-decision. However we know for $\NP$ that if $\SAT$ can be solved efficiently then a satisfying assignment can also be found efficiently \cite{AB09}. More recently, \cite{AHHN24} show a quantum oracle separation between the complexity classes $\QMA$ and $\BQP^{\uQMA}$. However for $\NP$ we know that if $\uSAT$ can be solved efficiently, then there is an efficient (randomized) algorithm for solving $\SAT$ \cite{VV85} (similarly for $\MA$ and $\QCMA$ \cite{ABOBS22}). \cite{AHHN24} also show a quantum oracle separation between $\QXC$ (quantum approximate counting) and $\QMA^{\QMA}$, even though classical approximate counting is in $\BPP^{\NP}$ \cite{Sto83}. Another line of work studied the differences between quantum and classical proofs specifically with the lens of no-cloning theorem for quantum states. In particular, \cite{NZ24, CKP24} study the cloneability of quantum proof states and they show that relative to quantum oracles, all of $\QMA$, $\cQMA$ (where the proof states are cloneable) and $\QCMA$ are distinct.

In quantum cryptography, \cite{Kre21} constructs a quantum oracle relative to which $\BQP = \QMA$ but cryptographic pseudorandom quantum states exists. See \cite{BCN25, CCS25, BMMMY25, CM24, GMMY24, CKMSTY25, GZ25, CGG24} for some other quantum oracle separations in quantum cryptography. \cite{GZ25} show conditions under which cryptographic separations in the Common Haar Random State Model can be lifted to quantum unitary oracle separations \cite{CCS25, AGL24a, AGL24b}. Many works in quantum cryptography have established quantum oracle separations even for primitives with classical inputs and outputs, in which case it is much more natural to establish black-box separations with quantum superposition access to a classical oracle.\footnote{In particular, such primitives are unlike the setting of Section 5.3 of \cite{CCS25}, since there is no well-defined notion of fully black-box constructions of primitive $\mathcal{Q}$ from unitary access to primitive $\mathcal{P}$ (and its inverse), as $\mathcal{P}$ measures its registers to produce the classical output.}

On the other hand, in order to make progress towards a classical oracle separation between $\QMA$ and $\QCMA$, \cite{NN24} introduced the distributional oracle model and showed a separation in this model. Later, \cite{LLPY24} construct a different distributional oracle separating $\QMA$ and $\QCMA$.

Since a classical oracle is a special case of quantum unitary oracles (diagonal unitaries with $\pm1$ entries), a quantum oracle separation is formally weaker than a classical oracle separation. So it is natural to wonder whether all the many quantum oracle separations mentioned previously establish a sufficiently strong relativization barrier. Indeed, in \cite{AK07}, the authors state the following: ``Currently, we do not know of any quantumly non-relativizing techniques that are not also classically non-relativizing.'' In simpler words, it means that for all known techniques, they are either sensitive to the presence of both classical and quantum oracles, or insensitive to the presence of both.

This claim was refuted by \cite{Aar09} to show that there is a technique not satisfying this: there is a complexity class containment which does not relativize with quantum oracles, but does relativize with a classical oracle. However, their result works only for a zero-error quantum complexity class called $\ZQEXP$ (zero-error quantum exponential time). This is not as natural a quantum complexity class to consider since quantum algorithms inherently have probabilistic outputs which can generally cause an error with some probability. In fact, \cite{Aar09} explicitly state it as an open problem: to show such a result for bounded-error quantum complexity classes.
%We resolve this open problem to show that there is a bounded-error quantum complexity class containment that does not relativize with a quantum oracle, but does relativize with classical oracles. We also show that similar results hold for the distributional oracle setting, which was not known previously to the best of our knowledge, even with zero-error complexity classes.

Finally, \cite{CKMSTY25} show some quantum oracle separations in cryptography, and their result holds even when the adversary is allowed access to the conjugate, transpose and inverse of the unitary oracle. They then state that: ``We are unaware of meaningful techniques that evade such a separation but do not evade a black-box separation with a classical oracle.'' It turns out that our result holds even when the $\polyQCPH$ verifier is given access to the conjugate, transpose and inverse of the unitary oracle. It is however unclear whether this is the case for the result of \cite{Aar09}.

\subsection{Our observations}\label{sec:obs}
\paragraph{Quantum oracles.} We show that there is a quantum oracle relative to which $\QMA$ is not contained in a class we call $\polyQCPH$. The class $\QCPH$ (introduced in \cite{GSSSY22}) informally consists of decision problems that can be solved by an efficient quantum verifier with constantly many alternating classical proofs. $\polyQCPH$ is defined similar to $\QCPH$, except the number of alternating classical proofs is allowed to grow polynomially with the length of the input. In the unrelativized world, and with respect to classical oracles, $\polyQCPH$ is equal to $\PSPACE$ (this is due to the $\PSPACE$-completeness of the $\TQBF$ problem), and it is a well-known result that $\QMA$ is contained in $\PSPACE$. Therefore, the containment $\QMA \subseteq \polyQCPH$, for bounded-error complexity classes, does not relativize quantumly, but does relativize classically. This resolves the open problem in \cite{Aar09}.
\begin{theorem}
    The following holds for the quantum complexity classes $\QMA$ and $\polyQCPH$:
    \begin{enumerate}
        \item $\QMA \subseteq \polyQCPH$
        \item For all classical oracles $O$, it holds that $\QMA^O \subseteq \polyQCPH^O$
        \item There exists a quantum unitary oracle $U$ such that $\QMA^U \not \subseteq \polyQCPH^U$
    \end{enumerate}
\end{theorem}
The first two parts of the claim hold because $\polyQCPH$ is equal to $\PSPACE$ with respect to classical oracles and it is a well-known result that $\QMA$ is contained in $\PSPACE$ \cite{MW05}. However, $\PSPACE$ has no interpretation in the quantum oracle model since it is a classical complexity class. Our main contribution is to introduce the quantum class $\polyQCPH$, which does have a meaningful interpretation in the quantum oracle model. Moreover, unlike $\PSPACE$ which is a zero-error class, $\polyQCPH$ is a bounded-error quantum complexity class and it also generalizes $\QCMA$. With this interpretation of $\PSPACE$ in the quantum oracle model, we can show our main result. The underlying technique that is not quantumly-relativizing is the same as that of \cite{Aar09}, that is, there is no longer a way to simulate a quantum algorithm by keeping track of its amplitudes. Our contribution lies in identifying the right bounded-error complexity class to establish this result.

\noindent It does not make sense to talk about $\PSPACE$ relative to quantum oracles, because $\PSPACE$ machines cannot access arbitrary quantum oracles. Instead one could talk about the class $\BQPSPACE$, which captures polynomial-space quantum computation with bounded error. In the unrelativized world, $\BQPSPACE$ is equal to $\PSPACE$. With respect to quantum oracles, the containment $\polyQCPH$ $\subseteq \BQPSPACE$ still holds, and so does $\QMA \subseteq \BQPSPACE$, so  we must have $\BQPSPACE \not\subseteq \polyQCPH$ in this setting. Therefore, while $\BQPSPACE$ is also an analogue of $\PSPACE$ in the quantum oracle model, it is not the right choice for our result. One should also note that $\BQPSPACE$ may not be contained in $\polyQCPH$ with respect to classical oracles either; $\PSPACE \subseteq \polyQCPH$ holds for classical oracles, but $\BQPSPACE = \PSPACE$ seems not to.

\noindent The quantum oracle with which we show $\QMA \not\subseteq \polyQCPH$ is the same one used in \cite{AK07} to separate $\QMA$ and $\QCMA$. We think this demonstrates that this oracle is powerful when separating quantum and classical resources in particular, and quantum oracles in general must be used with caution when attempting such separations.

\paragraph{Distributional oracles.} We show that there is a distributional oracle relative to which $\QMA$ is not contained in $\polyQCPH$. In particular, for classical oracles, $\PSPACE = \polyQCPH$ but this is not the case for distributional oracles. This demonstrates a bounded-error complexity class containment that holds for classical oracles but not distributional oracles; to the best of our knowledge, no such result was previously known even for zero-error classes.
\begin{theorem}
    There exists a distributional oracle $D$ such that $\QMA^D \not \subseteq \polyQCPH^D$.
\end{theorem}
The distributional oracle with which we show $\QMA \not\subseteq \polyQCPH$ is the same one used in \cite{LLPY24} to separate $\QMA$ and $\QCMA$. We think therefore that distributional oracles must also be used with caution when attempting separations.

\noindent Our findings are summarized in Figure~\ref{fig:containments}.

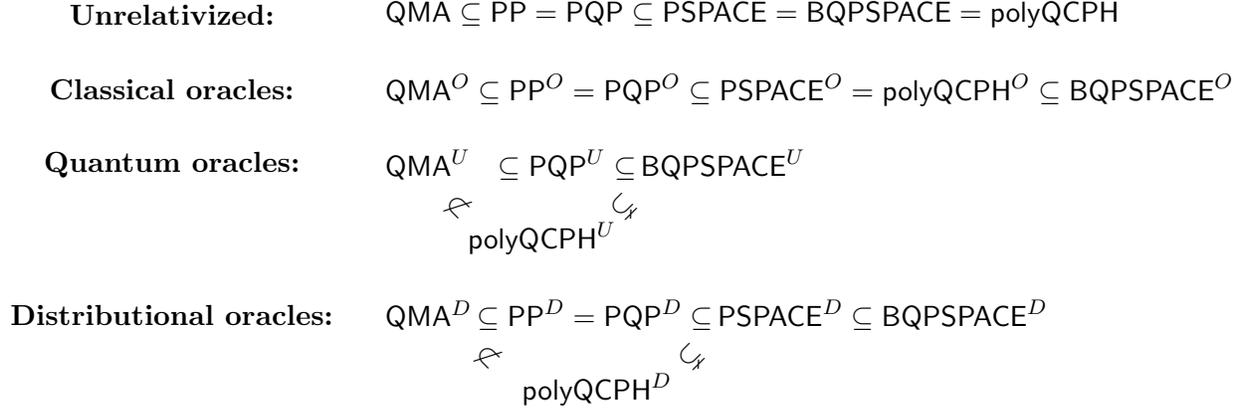
\begin{figure}[!ht]
\begin{tikzpicture}[every node/.style={align=left}]

% Unrelativized
\node at (0, 0) {\textbf{Unrelativized:}};
\node[right=2.7cm of {(0, 0)}] {$\QMA \subseteq \PP = \PQP \subseteq \PSPACE = \BQPSPACE = \polyQCPH$};

% w.r.t classical oracles
\node at (0,-1) {\textbf{Classical oracles:}};
\node[right=2.7cm of {(0, -1)}] {$ \QMA^O \subseteq \PP^O = \PQP^O \subseteq \PSPACE^O = \polyQCPH^O \subseteq \BQPSPACE^O$};

% w.r.t quantum oracles
\node at (0,-2) {\textbf{Quantum oracles:}};
\node[right=2.7cm of {(0,-2)}] (QMAU) {$ \QMA^U $};
\node [right=0.1cm of QMAU] {$\subseteq \PQP^U \subseteq$};
\node[right=2cm of QMAU] (BQPSPACEU) {$ \BQPSPACE^U $};
\node[below right=0.3cm and -0.3cm of QMAU] (polyQCPH) {$\polyQCPH^U $};
\node[rotate=-45, text=red] (nsubset) at (3.8,-2.55) {$\not\subset$};
\node[rotate=45] (subset) at (6,-2.55) {$\subsetneq$};

% w.r.t distributional oracle
\node at (0,-4) {\textbf{Distributional oracles:}};
\node[right=2.7cm of {(0,-4)}] (QMAD) {$ \QMA^D $};
\node [right=-0.2cm of QMAD] {$\subseteq \PP^D = \PQP^D \subseteq$};
\node[right=3cm of QMAD] (BQPSPACED) {$ \PSPACE^D \subseteq \BQPSPACE^D$};
\node[below right=0.3cm and 0.4cm of QMAD] (polyQCPH) {$\polyQCPH^D $};
\node[rotate=-45, text=red] at (4.2,-4.55) {$\not\subset$};
\node[rotate=45] at (6.9,-4.55) {$\subsetneq$};

\end{tikzpicture}
\caption{Complexity class containments in the unrelativized setting, and with respect to classical, quantum and distributional oracles. All of the containments shown here were known previously in the unrelativized setting \cite{MW05, Wat09, Wat03, AB09}; we verify that most of the containments hold in the quantum and distributional oracle settings, \emph{except} $\QMA \subseteq \polyQCPH$, for which we show a separation, highlighted in \textcolor{red}{red}.}
\label{fig:containments}
\end{figure}

\section{Definitions and Preliminaries}

\subsection{Complexity Classes}\label{sec:defs}

\begin{definition}[True Quantified Boolean Formula ($\TQBF$)]
    The language $\TQBF$ is defined as the set of quantified Boolean formula which are true, where a quantified Boolean formula $\psi$ is the following:
    \begin{align*}
        \psi = Q_1x_1Q_2x_2\ldots Q_nx_n \phi(x_1,x_2,\ldots x_n)
    \end{align*}
    where $Q_i$ are the quantifiers $\exists$ or $\forall$, and $\phi$ is a Boolean formula of size $m$. The size of $\psi$ is then defined to be $n+m$.
\end{definition}

\begin{definition}[Quantum Turing Machine \cite{Wat09}]\label{def:qtm}
    A quantum Turing machine (QTM) consists of three tapes: a read-only input tape, a classical work tape and a quantum work tape with qubits initialized to the $\ket{0}$ state. It has five tape heads, one each for the input and classical work tape, and three for the quantum tape. In every step of the computation, the Turing machine can read an input bit or perform one step on the classical tape, or apply one quantum gate (from Toffoli, Hadamard or phase-shift gates) on the quantum tape or perform a single-qubit measurement on the quantum tape.\\
    The runtime of the QTM is the number of steps taken by the machine, and the space used by the QTM is the number of squares used on the classical work tape plus the number of qubits used on the quantum work tape.
\end{definition}

\begin{definition}[$\BQPSPACE$ \cite{Wat09}]
    A promise problem $L = (L_{yes}, L_{no})$ is in $\BQPSPACE$ if there is a quantum Turing machine $T$ (as per Definition \ref{def:qtm}) and a polynomial $p: \N \rightarrow \N$ such that $T$ decides $L$ using $p(n)$ space on inputs of length $n$, with bounded error. That is,
    \begin{align*}
        x \in L_{yes} &\Rightarrow \Pr[T(x) = 1] \geq \frac{2}{3} \\
        x \in L_{no} &\Rightarrow \Pr[T(x) = 1] \leq \frac{1}{3}
    \end{align*}
\end{definition}
\noindent The classical analogue of $\BQPSPACE$ is $\PSPACE$ where the Turing machines described above do not have a quantum tape. It follows from definition that $\PSPACE^O \subseteq \BQPSPACE^O$ for every classical and distributional oracle. It is known that $\PSPACE = \BQPSPACE$ (shown in \cite{Wat03}). However, the proof of this result might not relativize for classical oracles according to our understanding.

\begin{definition}[$\PQP$]
    A promise problem $L = (L_{yes}, L_{no})$ is in $\PQP$ if there is a polynomial time uniform family of quantum circuits $\{V_n\}_{n \in \N}$ and polynomials $p,q: \N \rightarrow \N$ such that $V_n$ decides inputs of length $n$ using $p(n)$ time, with probability of error at most $1/2$. That is,
    \begin{align*}
        x \in L_{yes} &\Rightarrow \Pr[V_n(x) = 1] \geq \frac{1}{2} + \frac{1}{2^{q(n)}} \\
        x \in L_{no} &\Rightarrow \Pr[V_n(x) = 1] \leq \frac{1}{2}
    \end{align*}
\end{definition}
\noindent It is known (shown in \cite{Wat09}) that $\PQP = \PP$, where $\PP$ is the class of problems which can be solved by a classical polynomial time randomized algorithm with probability of error at most $1/2$. It can be verified that this result relativizes with respect to classical and distributional oracles.

\begin{definition}[$\QMA$]
    A promise problem $L = (L_{yes}, L_{no})$ is in $\QMA$ if there exists a poly-bounded function $p:\mathbb{N}\mapsto\mathbb{N}$ and a poly-time uniform family of quantum circuits $\{V_n\}_{n \in \mathbb{N}}$ such that for every $n$-bit input $x$, $V_n$ takes in a quantum proof $\ket{\psi}$ of $p(n)$ qubits and outputs a single bit, such that:
    \begin{align*}
        x \in L_{yes} &\Rightarrow \exists \ket{\psi}~~~~~\Pr[V_n(x, \ket{\psi}) = 1] \geq \frac{2}{3} \\
        x \in L_{no} &\Rightarrow \forall \ket{\psi}~~~~~\Pr[V_n(x,\ket{\psi}) = 1] \leq \frac{1}{3}
    \end{align*}
\end{definition}

\begin{definition}[$\QCSigma_i$]\label{def:QCSigmam}
  Let $L=(L_{yes},L_{no})$ be a promise problem. We say that $L$ is in $\QCSigma_i$ if there exists a poly-bounded function $p:\mathbb{N}\mapsto\mathbb{N}$ and a poly-time uniform family of quantum circuits $\{V_n\}_{n \in \mathbb{N}}$ such that for every $n$-bit input $x$, $V_n$ takes in classical proofs ${y_1}\in \set{0,1}^{p(n)}, \ldots, {y_i}\in \set{0,1}^{p(n)}$ and outputs a single bit, such that:
  \begin{itemize}
    \item Completeness: $x\in L_{yes}$ $\Rightarrow$ $\exists y_1 \forall y_2 \ldots Q_i y_i$ s.t. $\operatorname{Prob}[V_n \text{ accepts } (y_1, \ldots, y_i)] \geq \frac{2}{3}$.
    \item Soundness: $x\in L_{no}$ $\Rightarrow$ $\forall y_1 \exists y_2 \ldots \overline{Q}_i y_i$ s.t. $\operatorname{Prob}[V_n \text{ accepts } (y_1, \ldots, y_i)] \leq \frac{1}{3}$.
  \end{itemize}
  Here, $Q_i$ equals $\exists$ when $i$ is odd and equals $\forall$ otherwise and $\overline{Q}_i$ is the complementary quantifier to $Q_i$.
\end{definition}

\noindent Note that the first level of this hierarchy corresponds to $\QCMA$. The complement of the $i^{\mathrm{th}}$ level of the hierarchy, $\QCSigma_i$, is the class $\QCPi_i$ defined next.

\begin{definition}[$\QCPi_i$]\label{def:QCPim}
  Let $L=(L_{yes},L_{no})$ be a promise problem. We say that $L \in \QCPi_i$ if there exists a polynomially bounded function $p:\mathbb{N}\mapsto\mathbb{N}$ and a poly-time uniform family of quantum circuits $\{V_n\}_{n \in \mathbb{N}}$ such that for every $n$-bit input $x$, $V_n$ takes in classical proofs ${y_1}\in \set{0,1}^{p(n)}, \ldots, {y_i}\in \set{0,1}^{p(n)}$ and outputs a single bit, such that:
  \begin{itemize}
    \item Completeness: $x\in L_{yes}$ $\Rightarrow$ $\forall y_1 \exists y_2 \ldots Q_i y_i$ s.t. $\operatorname{Prob}[V_n \text{ accepts } (y_1, \ldots, y_i)] \geq \frac{2}{3}$.
    \item Soundness: $x\in L_{no}$ $\Rightarrow$ $\exists y_1 \forall y_2 \ldots \overline{Q}_i y_i$ s.t. s.t. $\operatorname{Prob}[V_n \text{ accepts } (y_1, \ldots, y_i)] \leq \frac{1}{3}$.
  \end{itemize}
  Here, $Q_i$ equals $\forall$ when $i$ is odd and equals $\exists$ otherwise, and $\overline{Q}_i$ is the complementary quantifier to $Q_i$.
\end{definition}

\noindent Now the corresponding quantum-classical polynomial hierarchy is defined as follows.

\begin{definition}[Quantum-Classical Polynomial Hierarchy ($\QCPH$)]\label{def:QCPH}
  \begin{align*}
        \QCPH = \bigcup_{m \in \mathbb{N}} \; \QCSigma_i = \bigcup_{m \in \mathbb{N}} \; \QCPi_i.
    \end{align*}
\end{definition}

\noindent Note that $\QCPH$ is the quantum-classical analogue of the classical complexity class $\PH$, where the verifier is a classical deterministic algorithm. The first level of $\PH$ is the well-known complexity class $\NP$. It follows from definition that $\PH \subseteq \QCPH$ and this is true relative to all classical and distributional oracles. See \cite{GSSSY22}, \cite{AGKR24} and \cite{AB24} for other properties of $\QCPH$. Note also that $\QCPH$ ($\PH$) consist of problems which can be decided by a quantum (classical) verifier receiving a constant number of alternatively quantified proofs. In particular, the number of proofs that the verifier receives is not allowed to increase with increase in the input length (though the size of the proofs is allowed to increase). If we do allow the number of proofs to grow, say, polynomially with the length of the input, then we get a different complexity class, which we will call $\polyQCPH$. We will call the analogous classical deterministic complexity class as $\polyPH$, though we don't define it formally. Later we show that $\polyPH = \polyQCPH$ with respect to all classical oracles. We need the class $\polyQCPH$ (instead of the previously defined class $\QCPH$) because we need a complexity class that is equal to $\PSPACE$ in the unrelativized world.

\begin{definition}[$\polyQCSigma_i$]\label{def:polyQCSigmam}
  Let $L=(L_{yes},L_{no})$ be a promise problem. We say that $L$ is in $\polyQCSigma_i$ if there exists a poly-bounded function $p:\mathbb{N}\mapsto\mathbb{N}$ and a poly-time uniform family of quantum circuits $\{V_n\}_{n \in \mathbb{N}}$ such that for every $n$-bit input $x$, $V_n$ takes in at most $n^i$ classical proofs ${y_1}\in \set{0,1}^{p(n)}, \ldots, {y_{n^i}}\in \set{0,1}^{p(n)}$ and outputs a single bit, such that:
  \begin{itemize}
    \item Completeness: $x\in L_{yes}$ $\Rightarrow$ $\exists y_1 \forall y_2 \ldots Q_{n^i} y_{n^i}$ s.t. $\operatorname{Prob}[V_n \text{ accepts } (y_1, \ldots, y_{n^i})] \geq \frac{2}{3}$.
    \item Soundness: $x\in L_{no}$ $\Rightarrow$ $\forall y_1 \exists y_2 \ldots \overline{Q}_{n^i} y_{n^i}$ s.t. $\operatorname{Prob}[V_n \text{ accepts } (y_1, \ldots, y_{n^i})] \leq \frac{1}{3}$.
  \end{itemize}
  Here, $Q_{n^i}$ equals $\exists$ when ${n^i}$ is odd and equals $\forall$ otherwise and $\overline{Q}_{n^i}$ is the complementary quantifier to $Q_{n^i}$.
\end{definition}

\begin{definition}[$\polyQCPi_i$]\label{def:polyQCPim}
  Let $L=(L_{yes},L_{no})$ be a promise problem. We say that $L \in \polyQCPi_i$ if there exists a polynomially bounded function $p:\mathbb{N}\mapsto\mathbb{N}$ and a poly-time uniform family of quantum circuits $\{V_n\}_{n \in \mathbb{N}}$ such that for every $n$-bit input $x$, $V_n$ takes in at most $n^i$ classical proofs ${y_1}\in \set{0,1}^{p(n)}, \ldots, {y_{n^i}}\in \set{0,1}^{p(n)}$ and outputs a single bit, such that:
  \begin{itemize}
    \item Completeness: $x\in L_{yes}$ $\Rightarrow$ $\forall y_1 \exists y_2 \ldots Q_{n^i} y_{n^i}$ s.t. $\operatorname{Prob}[V_n \text{ accepts } (y_1, \ldots, y_{n^i})] \geq \frac{2}{3}$.
    \item Soundness: $x\in L_{no}$ $\Rightarrow$ $\exists y_1 \forall y_2 \ldots \overline{Q}_{n^i} y_{n^i}$ s.t. $\operatorname{Prob}[V_n \text{ accepts } (y_1, \ldots, y_{n^i})] \leq \frac{1}{3}$.
  \end{itemize}
  Here, $Q_{n^i}$ equals $\forall$ when ${n^i}$ is odd and equals $\exists$ otherwise, and $\overline{Q}_{n^i}$ is the complementary quantifier to $Q_{n^i}$.
\end{definition}

\begin{definition}[$\polyQCPH$]\label{def:polyQCPH}
  \begin{align*}
        \polyQCPH = \bigcup_{i \in \mathbb{N}} \; \polyQCSigma_i = \bigcup_{i \in \mathbb{N}} \; \polyQCPi_i.
    \end{align*}
\end{definition}
Note that we later make use of the classical complexity class $\polyPH$, which is defined similarly to $\polyQCPH$, except that the verifier is classical instead of quantum.
\subsection{Other Preliminaries}

\begin{definition}[$p$-uniform probability measure]
    Let $\mu$ be the uniform probability over $N$-dimensional pure states. For a given $p \in [0,1]$, a probability measure $\sigma$ is a $p$-uniform measure if it can be obtained from $\mu$ by conditioning on some event occurring with probability at least $p$.
\end{definition}

\begin{lemma}[Geometric Lemma \cite{AK07}]\label{lem:geometriclemma}
    Given a $p$-uniform probability measure $\sigma$ over $N$-dimensional pure states, and a density matrix $\rho$,
    \begin{align*}
        \E_{\ket{\psi} \sim \sigma}[\bra{\psi}\rho\ket{\psi}] = O\left(\frac{1+\log(1/p)}{N}\right)
    \end{align*}
\end{lemma}

\begin{proposition}[Chernoff Bound]\label{prop:chernoff}
    Let $\mathbf{X_1}, \cdots, \mathbf{X_n}$ be $n$ independent random variables between $0$ and $1$. If $\mathbf{X} = \sum_{i=1}^n \mathbf{X_i}$ and $\mu = \E[\mathbf{X}]$, then for any $\delta > 0$
    \begin{align*}
        \Pr[\mathbf{X} \geq \mu + t] &\leq e^{-\frac{2t^2}{n}} \\
        \Pr[\mathbf{X} \leq \mu - t] &\leq e^{-\frac{2t^2}{n}}
    \end{align*}
\end{proposition}

\section{Classical Oracles}
In this section, we show that $\PSPACE = \polyQCPH = \polyPH$ with respect to every classical oracle $O$. We start by describing the proof of $\TQBF$-completeness for $\PSPACE$, from \cite{AB09}.

\begin{lemma}[Theorem 4.11 of \cite{AB09}]\label{lem:tqbfpspacecomp}
    The problem $\TQBF$ is complete for $\PSPACE$. This holds relative to all classical oracles $O$.
\end{lemma}
\begin{proof}
    Consider $L \in \PSPACE$ and let $M$ be the corresponding Turing Machine. Let $G_{M,x}$ be the configuration graph of $M$ on input $x$. Then $M$ accepts $x$ iff there is a path from $C_{init}$ to $C_{acc}$ in the configuration graph $G_{M,x}$. Suppose $M$ uses space $m = p(|x|)$ on input $x$, where $p: \N \rightarrow \N$ is a polynomial. We will convert this problem of finding a path between $C_{init}$ and $C_{acc}$ to a quantified Boolean formula $\psi$ which is true iff such a path exists. This construction is done inductively, and we define $\psi_i(C_1, C_2)$ to be the quantified Boolean formula which is true iff there is a path of length at most $2^i$ between vertices $C_1$ and $C_2$ of $G_{M,x}$. Then we want to set $\psi = \psi_m(C_{init}, C_{acc})$. We know that $\psi_0$ can be written as a Boolean formula of size $O(m)$, using a construction similar to Cook-Levin theorem. Then we can define $\psi_i$ as follows:
    \begin{align*}
        \psi_i(C_1, C_2) := \exists C \forall D_1 \forall D_2 [(D_1 = C_1 \wedge D_2 = C) \vee (D_1 = C \wedge D_2 = C_2)] \Rightarrow \psi_{i-1}(D_1, D_2)
    \end{align*}
    The formula above encodes the fact that there is a path of length at most $2^i$ between $C_1$ and $C_2$ iff there is a vertex $C$ such that there is a path of length at most $2^{i-1}$ between $C_1$ and $C$ and between $C$ and $C_2$. Then the formula $\psi_m$ has size $O(m^2)$. We can also check that each of these steps will work in the presence of an arbitrary classical oracle, where note it is important that the quantified variables $C, D_1, D_2$ and the Boolean formulas $\psi_i$ are allowed to depend on the oracle $O$.
\end{proof}

\begin{lemma}
    For every classical oracle $O$, $\PSPACE^O = \polyPH^O = \polyQCPH^O$.
\end{lemma}
\begin{proof}
    It is clear from definition that $\polyPH^O \subseteq \polyQCPH^O \subseteq \PSPACE^O$ where for the second containment we use the fact that a $\PSPACE^O$ machine can iterate over all possible settings of alternatively quantified proofs and estimate the acceptance probability of the resulting verifier. In particular, we can think of the action of verifier $V^O$ on input $x$ as follows:
    \begin{align*}
        V^O(x) = V_{p(n)}OV_{p(n)-1}\ldots OV_1 \ket{x0^{q(n)}}
    \end{align*}
    where $q(n)$ is the number of ancilla qubits used by $V^O$. Then a classical Turing machine $M^O$ can compute the probability that the output on measuring $V^O(x)$ is $y$, by computing the following sum:
    \begin{align*}
        \sum_{z_1, z_2, \ldots z_{2p(n)-2}} \bra{y}V_{p(n)}\ket{z_{2p(n)-2}}\bra{z_{2p(n)-2}}O_D\ket{z_{2p(n)-3}} \ldots \bra{z_2}O_D\ket{z_1} \bra{z_1}V_1\ket{x0^{q(n)}}
    \end{align*}
    This sum can be computed in polynomial space, since $M$ only needs to store a partial sum and $2p(n)-1$ values for computing the products in each sum. Note that to compute these products, $M$ only needs to make classical queries to $O$ because $z_i$ are classical strings. From Lemma \ref{lem:tqbfpspacecomp}, we know that $\PSPACE^O \subseteq \polyPH^O$ since a $\polyPH^O$ machine can clearly decide a $\TQBF$ predicate (dependent on the oracle $O$). Hence the claim follows.
\end{proof}

\section{Quantum Oracles}

\subsection{\texorpdfstring{$\QMA$}{QMA} is contained in \texorpdfstring{$\PQP$}{PQP}}
In this section we show that relative to every quantum oracle $U$, $\QMA^U \subseteq \PQP^U$. This proof is the same as the proof by \cite{MW05}.

\begin{lemma}[Theorem 3.4 of \cite{MW05}]\label{lem:qmainpqp}
    For every quantum oracle $U$, $\QMA^U \subseteq \PQP^U$.
\end{lemma}
\begin{proof}
    Let $L \in \QMA^U$ and $V^U$ be the corresponding verifier. Suppose $V^U$ takes a quantum proof of length $p(n)$ on inputs of length $n$. Using strong error-reduction for $V^U$ from \cite{MW05}, we will assume that $V^U$ makes an error of at most $1/2^{r(n)}$, where $r(n) = p^2(n)$. Define the verifier $W^U$, which runs the verifier $V^U$ by fixing the proof to be the maximally mixed state on $p(n)$ qubits. That is, on input $x$ of length $n$, $W^U(x) = V^U\left(x, \frac{I_{2^{p(n)}}}{2^{p(n)}}\right)$. Then using the assumption on $V^U$, we have the following:
    \begin{align*}
        x \in L_{yes} &\Rightarrow \Pr[W^U(x) = 1] \geq \frac{1}{2^{p(n)}}\left(1-\frac{1}{2^{r(n)}}\right) \\
        x \in L_{no} &\Rightarrow \Pr[W^U(x) = 1] \leq \frac{1}{2^{r(n)}}
    \end{align*}
    Thus $W^U$ is a $\PQP^U$ verifier for $L$ for the given choice of $r(n)$. Therefore, $L \in \PQP^U$.
\end{proof}

\begin{lemma}\label{lem:pqpinbqpspace}
    For every quantum oracle $U$, $\PQP^U \subseteq \BQPSPACE^U$.
\end{lemma}
\begin{proof}
    Let $L \in \PQP^U$ and $V^U$ be the corresponding verifier. Suppose the verifier $V^U$ accepts with probability at least $\frac{1}{2} + \frac{1}{2^{p(n)}}$ on an input $x$ of length $n$ if $x \in L_{yes}$ and with probability at most $\frac{1}{2}$ otherwise. Then define a verifier $W^U$ which on input $x$ repeats $V^U(x)$ independently $c2^{4p(n)}$ times for a sufficiently large constant $c$ and accepts iff the fraction of accepting iterations is at least $\frac{1}{2} + \frac{1}{2^{p(n)}} - \frac{1}{2^{2p(n)}}$. Using Chernoff bound (Proposition \ref{prop:chernoff}), $W^U$ makes an error with probability at most $1/3$ if $c$ is a sufficiently large constant. Further, $W^U$ can be run in polynomial space, since it can reuse the same qubits for each iteration of running $V^U(x)$ (which takes polynomially many qubits) and store a counter for the number of accepting iterations on the classical tape, which take $O(p(n))$ classical bits. Hence, $L \in \BQPSPACE^U$.
\end{proof}

\paragraph{$\TQBF$ completeness fails for quantum oracles.}In the classical setting, it is known that if we allow the number of alternating proofs in $\PH$ to grow polynomially with the length of the input, then we get the same class as $\PSPACE$. This follows from Lemma \ref{lem:tqbfpspacecomp}. However, in the case of quantum oracles, $\PSPACE$ and $\TQBF$ do not make sense anymore (note that $\TQBF$ may not be complete for $\BQPSPACE$ even for classical oracles). If we try to do a similar proof as the classical one, but now for $\BQPSPACE$ and $\polyQCPH$, it also does not seem to work. In particular, suppose the classical alternatively quantified proofs are supposed to capture a path in the configuration graph of the classical tape of a $\BQPSPACE$ machine $Q$. Since $Q$ is a $\BQPSPACE$ machine, it can potentially run for an exponential amount of time, provided it does not use more than polynomial space. Therefore, we can't defer measurements in $Q$ by adding ancilla qubits. In particular, the future steps of the classical tape can depend on these intermediate measurement outcomes. Therefore, a $\polyQCPH$ predicate may not capture this possibility, since all the quantum measurements in $\polyQCPH$ happen after all the classical proofs have been specified. Therefore, while $\polyQCPH^U \subseteq \BQPSPACE^U$, the reverse containment does not seem to hold.

\subsection{Oracle Separation between \texorpdfstring{$\QMA$}{QMA} and \texorpdfstring{$\polyQCPH$}{polyQCPH}}
In this section we show that the quantum oracle separation between $\QMA$ and $\QCMA$ established in \cite{AK07} extends to an oracle separation between $\QMA$ and $\QCPH$. More strongly, it extends to a separation between $\QMA$ and $\polyQCPH$, thus also giving a quantum oracle separation between $\PQP$ and $\polyQCPH$.
\paragraph{\cite{AK07} Quantum oracle problem.} Given either the oracle $U_\psi = I - 2\ket{\psi}\bra{\psi}$ for some $n$-qubit Haar-random state $\ket{\psi}$ (YES case), or the identity operator $I$ on qubits (NO case), distinguish between the two cases. \\
\\
\noindent Given a query algorithm $A$, \cite{AK07} consider for every $n$-qubit pure state $\ket{\psi}$, the classical witness $w$ which maximizes the acceptance probability of $A$ on $U_\psi$. They use the witnesses to partition the space of $n$-qubit pure states, say the partition for witness $w$ is called $S(w)$, and they pick the witness $w$ with the partition of the largest size. Then they show that if $A$ makes a small number of queries, then $A(\cdot, w)$ can not distinguish between $U_\psi$ and $U = I$ with high probability over the choice of a uniformly random state $\ket{\psi}$ from $S(w)$. We present a slightly different formulation of the proof from \cite{AK07}, by considering for every $n$-qubit pure state $\ket{\psi}$, the classical witness $w$ which maximizes the gap between the acceptance probability of $A$ for $U_\psi$ and $U = I$.

\begin{theorem}[Modification of Theorem 3.3 of \cite{AK07}]\label{thm:lowerbound}
    Suppose we have a query algorithm $A$ with oracle access to an $n$-qubit unitary $U$, which is chosen from one of the two cases below:
    \begin{enumerate}
        \item $U = U_\psi$
        %There exists a state $\ket{\psi}$ such that $U\ket{\psi} = -\ket{\psi}$ and $U\ket{\phi} = \ket{\phi}$ for every state $\ket{\phi}$ that is orthogonal to $\ket{\psi}$
        \item $U = I$. 
    \end{enumerate}
    For a particular oracle $U_\psi$, let $w$ be the witness of length $m$ which maximizes the distinguishing advantage of $A$ between $U_\psi$ and $I$ (we call this good witness). For a given witness $w$, define $S(w)$ to be the set of states $\ket{\psi}$ such that $w$ is the good witness for $U_\psi$. Then there exists a witness $w^{\ast}$ such that with high probability over choice of states $\ket{\psi} \in S(w^{\ast})$, the algorithm $A$ (given any witness of length $m$) needs $\Omega\left(\sqrt{\frac{2^n}{m+1}}\right)$ queries to distinguish between $U_\psi$ and $I$ with bounded error.
\end{theorem}
\begin{proof}
    Suppose $A$ makes at most $T$ queries. We know by definition that $\{S(w)\}_w$ forms a partition of the set of $n$-qubit pure states, and we can pick a witness $w^{\ast}$ such that
    \begin{align*}
        \Pr_{\ket{\psi}\sim \mu}[\ket{\psi} \in S(w^{\ast})] \geq \frac{1}{2^m}
    \end{align*}
    For every state $\ket{\psi} \in S(w^{\ast})$, let $\ket{\Phi^\psi}$ be the final state of $A$ when $U = U_\psi$, and $\ket{\Phi^I}$ be the final state of $A$ when $U = I$. Then \cite{AK07} show, using Lemma \ref{lem:geometriclemma}, that
    \begin{align}\label{eqn:hybrid}
        \E_{\ket{\psi}\sim S(w^{\ast})}[||\ket{\Phi^\psi} - \ket{\Phi^I}||_2] \leq O\left(T \sqrt{\frac{m+1}{2^n}}\right).
    \end{align}
    Note that the above proof also works when the quantum algorithm is given query access not just to the unitary $U$, but also its conjugate, transpose and inverse. In particular, Equation \ref{eqn:hybrid} is obtained using the hybrid method, which essentially replaces one-by-one the queries of the quantum algorithm in the YES case with the $I$ oracle (NO case), while arguing that the behaviour of the algorithm does not change much. This process can also be done when queries to the conjugate, transpose or inverse of the unitary $U_{\psi}$ are present.
    Therefore, using Markov's inequality, for every query algorithm $A$ we have,
    \begin{align}
        \Pr_{\ket{\psi}\sim S(w^{\ast})}\left[||\ket{\Phi^\psi} - \ket{\Phi^I}||_2 \geq \omega\left(T \sqrt{\frac{m+1}{2^n}}\right)\right] \leq o(1).  \label{eq:p-uniform}
    \end{align}
    In particular, with probability at least $2/3$ over choice of states $\ket{\psi} \in S(w^{\ast})$, $A$ needs $T = \Omega\left(\sqrt{\frac{2^n}{m+1}}\right)$ queries to distinguish between $U_\psi$ and $U = I$ with constant bias.
\end{proof}

\begin{theorem}\label{thm:QMA-polyQCPH-U}
    There exists a quantum oracle $U$ such that $\QMA^U \not\subset \polyQCPH^U$.
\end{theorem}
\begin{proof}
    Let $U = \{U_n\}_n$ be a family of unitaries, where each unitary is chosen from one of the two cases from Theorem \ref{thm:lowerbound}. Define a language $L^U$ corresponding to the oracle $U$ as follows: $1^n \in L^U$ iff $U_n = U_\psi$ for some state $\ket{\psi}$. It is shown in \cite{AK07}, that $L^U \in \QMA^U$ for all quantum oracles $U$. \\
    \\
    \noindent Now we show that there exists a quantum oracle $U$ such that $L^U \notin \polyQCPH^U$. We order all the $\polyQCPH^U$ verifiers $A$ and polynomials $p, q: \N \rightarrow \N$ in some order $\{A_i, p, q\}_{i \in \N, p, q}$. We will assume without loss of generality that $p, q$ are of the form $n^k$ for $k \in \N$, so there are only countably many polynomials. Pick the verifier $A_i$, which runs in time $p(n)$ on inputs of length $n$. Consider the action of $A_i$ on a sufficiently long input of length $n_i$ (depending on $p,q$), where $n_i > N$, where we have already fixed the oracle on inputs upto length $N$. Suppose $A_{n_i}$ takes $q(n_i)$ alternatively quantified proofs each of length $q(n_i)$. Fix the witness $w^{\ast}$ (of length, $m = q^2(n_i)$) such that $S(w^{\ast})$ satisfies the following:
    \begin{align*}
        \Pr_{\ket{\psi}\sim \mu}[\ket{\psi} \in S(w^{\ast})] \geq \frac{1}{2^m}
    \end{align*}
    Then we know from Theorem \ref{thm:lowerbound} that with probability at least $2/3$ over choice of states $\ket{\psi} \in S(w^{\ast})$, $A_i(1^{n_i}, w^{\ast})$ fails to distinguish between $U_\psi$ and $U = I$ with constant bias, that is,
    \begin{align*}
        |\Pr[A_i^{U_\psi}(1^{n_i}, w^{\ast}) = 1] - \Pr[A_i^{U}(1^{n_i}, w^{\ast}) = 1]| \leq O\left(T \sqrt{\frac{m+1}{2^{n_i}}}\right) = o(1).
    \end{align*}
    since $T \leq p(n_i)$. Pick one such state $\ket{\psi} \in S(w^{\ast})$. Now we set $U_{n_i}$ as follows: if
    \begin{align*}
        \exists y_1 \forall y_2 \ldots Q_{q(n_i)}y_{q(n_i)} \Pr[A^{U_\psi}(1^{n_i}, y_1, y_2, \ldots, y_{q(n_i)}) = 1] > \frac{1}{2}
    \end{align*}
     then set $U_{n_i} = I_{n_i}$. Otherwise set $U_{n_i}$ to be $U_\psi$. Note that the distinguishing bias of $A$ for $U_{n_i}$ is highest for the witness $w^{\ast}$, so if this distinguishing bias is $o(1)$ then so is the bias for every witness. Therefore, in case 1, when $U_{n_i} = I_{n_i}$:
    \begin{align*}
        \exists y_1 \forall y_2 \ldots Q_{q(n_i)}y_{q(n_i)} \Pr[A^{U_{n_i}}(1^{n_i}, y_1, y_2, \ldots, y_{q(n_i)}) = 1] > \frac{1}{3}
    \end{align*}
    In case 2, when $U_{n_i} = U_\psi$:
    \begin{align*}
        \forall y_1 \exists y_2 \ldots \overline{Q}_{q(n_i)}y_{q(n_i)} \Pr[A^{U_{n_i}}(1^{n_i}, y_1, y_2, \ldots, y_{q(n_i)}) = 1] < \frac{2}{3}
    \end{align*}
    Therefore, $U$ is set such that $A_i$ (running in time $p$ and taking witnesses of length $q$) makes an error in deciding $L^U$ on the input of length $n_i$.
\end{proof}

\subsection{Discussion on Classical Oracle Separations}\label{sec:disc}
It's important to understand why a strategy like in the proof of Theorem~\ref{thm:QMA-polyQCPH-U} would not work for classical oracles, and we now provide some discussion on this. The particular focus will be on the full classical oracle separations of \cite{BHNZ25, BHV26}. In the \cite{AK07} quantum oracle problem, there is a single NO oracle $U=I$. Therefore for a given algorithm $A$, and any YES instance $U_\psi$, we can consider a witness $w$ which maximizes the distinguishing advantage of $A$ between $U_\psi$ and $I$. This can then be used to define the set $S(w)$ for every witness $w$ and it allows us to pick $w^{\ast}$, such that $S(w^{\ast})$ has large number of YES instances. We can now pick a YES instance from $S(w^{\ast})$, and no other witness can do better than $w^{\ast}$ to distinguish from the single NO instance. If there were multiple NO instances, then for any YES instance, there will be different witnesses $w$ maximizing the distinguishing advantage for different NO instances. In particular, it is possible that any given choice of proofs $y_1,\ldots, y_{q(n)}$ maximizes the distinguishing advantage for \emph{some} NO instance. It's not clear how to define the set $S(w)$ and consider $S(w^\ast)$ under these circumstances.

We believe it's simply not possible to get a classical oracle problem with a single NO instance for which a statement such as Theorem \ref{thm:lowerbound} holds. Indeed, this can be verified to not be true for the candidate classical oracle separations between $\QMA$ and $\QCMA$ in \cite{LLPY24,BK24,Zha25,LMY24, BHNZ25, BHV26} (note that \cite{BHNZ25, BHV26} show a full classical oracle separation between $\QMA$ and $\QCMA$). All of these have multiple NO instances. For these oracles, one could try to pick a $w$ which maximizes the distinguishing advantage between the YES instance and some fixed NO instance. We now discuss why this way of picking $S(w)$ cannot be combined with the proof techniques of \cite{LLPY24,BK24,Zha25,LMY24} to give a separation between $\QMA$ and $\polyQCPH$ (or even $\QCPH$ --- indeed we believe all of these problems are contained in $\PH$).\footnote{Note that none of these works prove a full unconditional separation between $\QMA$ and $\QCMA$ with a classical oracle. However, unlike the $\QMA$ vs $\QCMA$ candidate of \cite{AK07}, the candidates in these works cannot be extended to prove a $\QMA$ vs $\polyQCPH$ separation even in the limited or conditional settings that are considered in these works.} In all of these works, after fixing a witness $w$, the proof technique shows shows that there is a YES instance (or there are many YES instances) for which this $w$ is a valid witness, and a NO instance (or many NO instances) such that the algorithm $A(w)$ after fixing $w$ cannot tell the YES instance from the NO instance. However, the NO instance here might not be the same NO instance we maximized the gap in acceptance probability with to define $S(w)$ in the first place, so this strategy will not work.

One could instead consider a different way of defining $S(w)$, for example, for each YES instance, pick the $w$ that maximizes the gap between the acceptance probability of that YES instance, and the average acceptance probability of NO instances. Suppose we are working with two alternating proofs, so $w=y_1y_2$. Not every $y_2$ is a valid proof for every NO instance, so the expectation described above has to be over NO instances for which $y_2$ is a valid proof. So overall, for every YES instance we pick a proof $y_1$ and for every NO instance we pick a proof $y_2$ such that
\[ \left|\Pr\left[A^\text{yes}(y_1y_2)\right] - \E_{\text{no} \sim T(y_2)}\Pr\left[A^\text{no}(y_1y_2)\right]\right|\]
is maximized, where $T(y_2)$ is the set of NO instances for which $y_2$ is the picked proof. If $y_1$ and $y_2$ are both of length $p(n)$, we can pick $y_1^*$ and $y_2^*$ such that the set $S(y_1^*)$ of YES instances corresponding to $y_1^*$ is a $2^{-p(n)}$ fraction of all YES instances, and $T(y_2^*)$ is a $2^{-p(n)}$ fraction of all NO instances.\footnote{This way of picking the witness can be extended to when we have more than two alternating proofs as well --- we simply restrict to a fraction of YES or NO instances for each alternating proof.} One can then try to find YES instance in $S(y_1^*)$ and a NO instance in $T(y_2^*)$ (for which the acceptance gap with $S(y_1^*)$ is actually small --- but there may be many such instances) such that they cannot be distinguished by $A(y_1^*y_2^*)$. This strategy should fail for all of the classical candidates in \cite{LLPY24,BK24,Zha25,LMY24}. In the rest of this discussion, we provide a detailed explanation for why any such strategy fails for the candidates in \cite{BHNZ25, BHV26}. The rest of this section is skippable without loss of continuity, for readers who are not interested in a detailed explanation.

 \paragraph{\cite{BHNZ25} Classical Oracle Separation between $\QMA$ and $\QCMA$.} This work separates $\QMA$ and $\QCMA$ with respect to a classical oracle. The classical oracle consists of two functions $S, U: \{0,1\}^n \rightarrow \{0,1\}$. For the rest of the discussion we will think of the functions $S, U$ as specifying two subsets $S, U \subseteq \{0,1\}^n$ ($x \in S$ iff $S(x) = 1$), and we define $\Pi_S, \Pi_U$ to be the projectors onto the subspaces specified by the subsets $S, U$. Then the goal is to distinguish between the following two cases:
 \begin{itemize}
     \item YES case: $\exists~\ket{\psi}$: $||\Pi_UH^{\otimes n}\Pi_S\ket{\psi}||^2 \geq a$
     \item NO case: $\forall~\ket{\psi}$: $||\Pi_UH^{\otimes n}\Pi_S\ket{\psi}||^2 \leq b$
 \end{itemize}
 They show that for a certain distribution $\mathsf{Strong}$ over $(S,U)$, if there is a $\QCMA$ algorithm that can distinguish between the YES and NO cases successfully, then there is an efficient quantum sampler which when given oracle access to $U$, can generate many distinct samples from $S$ with reasonable probability (which will later be used to derive a contradiction). The intuition for this is the fact that if a certain classical proof can help distinguish between the two cases, then it can be reused multiple times to generate many samples from $S$, and that should not be possible (they show in the paper that getting many samples from $S$ by looking at $U$ should not be possible with high probability). The distribution $\mathsf{Strong}$ has the property that if $\Delta \subseteq S$ has a sufficiently small size, then $(\Delta, U)$ forms a NO instance. The conversion from the $\QCMA$ algorithm to the sampler then works as follows: Guess the classical witness $w$, and set $\Delta = \emptyset$. Then pick the $j^{th}$ query to $S$ uniformly at random, run the algorithm (with $O_{\Delta}$ instead of $O_S$) upto the $j^{th}$ query, and measure the state to get a string $x$. This should be a string in $S\setminus \Delta$ if the algorithm distinguishes between $(S,U)$ and $(\Delta, U)$. Output $x$ and update $\Delta = \Delta \cup \{x\}$. Repeat the same process to get the desired number of samples.

 Note that the above conversion crucially uses the fact that the witness $w$ guessed by the sampler, if correct, can help the algorithm distinguish between $(S,U)$ and $(\Delta, U)$ for any $\Delta \subseteq S$ such that $(\Delta, U)$ is a NO instance. Thus, if the sampler guesses the correct witness, then it can keep reusing it for each iteration and get a new sample each time. This means that the sampler pays an initial cost of loss in probability for guessing the witness, but does not incur this cost for every iteration. However, this is not the case if we had a $\QCPH$ algorithm instead. To illustrate this, consider $\QCSigma_2$: in the NO case, we have the promise that $\forall w_1 \exists w_2$ such that the algorithm accepts with low probability. In particular, the witness $w_2$ can be different for different choices of $\Delta$, which means that the sampler as described above, will need to guess the witness $w_2$ again for every iteration. This means that the sampler will have a very small success probability (hence we do not get a similar contradiction as the $\QCMA$ case).

\paragraph{\cite{BHV26} Simpler classical Oracle Separation between $\QMA$ and $\QCMA$.} \cite{BHV26} give a classical oracle separation between QMA and QCMA with a simpler proof than in \cite{BHNZ25}. Their oracle is a modification of the oracles in \cite{LLPY24} and \cite{BK24}, and their result basically subsumes the results of those works. The classical oracle $O[H,E]$ in \cite{BHV26} is parametrized by a uniformly random function $H:\{0,1\}^{\Theta(n\log n)} \to \{0,1\}$, and a subset $E \subseteq \{0,1\}^n\times\{0,1\}^{\Theta(n^2)}$. Each function $H$ defines a relation $R_H$ on $\{0,1\}^n\times\{0,1\}^{\Theta(n^2)}$, and the entries of the oracle $O[H,E]$ are indexed by $(x,u) \in \{0,1\}^n\times\{0,1\}^{\Theta(n^2)}$. The entries are given by
\[ O[H,E](x,u) = \begin{cases} 1 & \text{ if } (x,u) \in R_H \land (x,v) \in E \\ 0 & \text{ otherwise,}\end{cases} \]
and the goal is to distinguish between $E = \{0,1\}^n\times\{0,1\}^{\Theta(n)}$ (YES case) and $E \subseteq F\times\{0,1\}^{\Theta(n)}$, where $F \leq 2^n/3$ (NO case).
The core proof idea is similar to \cite{BHNZ25}: they use a QCMA algorithm to produce many $(x,v)$ pairs in the relation $R_H$ without making any queries at all. They then use list decoding properties of the code used to define $R_H$ to show that producing many such pairs in $R_H$ without any queries is only possible with very small probability. The reduction from the QCMA algorithm to the sampler for $R_H$, as before, guesses a QCMA witness for the yes instance corresponding to the unknown $H$ (there is only a single yes instance for every $H$), runs the QCMA algorithm on the NO instance with $E=\emptyset$ (which does not depend on $H$ at all), measures the query register on a random query, and then updates $E$ with the measurement result. As long as this is not run for too many iterations, the size of $E$ in each iteration will be small enough for the resulting oracle to be a NO instance, and a new $(x,v)$ pair in $R_H$ will be produced with decent probability.

The problem in \cite{BHNZ25} is likely not in $\QCPH$, but this problem is in $\PH$, and in fact in $\AM$. So the proof technique should not generalize to QCPH algorithms. We can see that this is indeed the case for similar reasons as for \cite{BHNZ25}: the technique relies on being able to use the same YES instance witness to distinguish many NO instances, and this is no longer true for $\QCPH$ witnesses.

\section{Distributional Oracles}
The distributional oracle model was introduced by \cite{NN24}, with the aim of making progress towards showing a classical oracle separation between $\QMA$ and $\QCMA$. In this model, there is a family of YES distributions $\{D_Y\}_Y$ and a family of NO distributions $\{D_N\}_N$. Each distribution $D_Y$ (or $D_N$) has support on, say, $M$ elements. The oracle $O$ is sampled from one of these families of distributions and the aim of an algorithm is to tell which family of distributions the oracle is sampled from, potentially with the help of a prover who can send a classical witness. The crucial difference from classical oracle setup is that the prover only knows the distribution $D_Y$ (or $D_N$) from which $O$ is sampled, and not the explicit sample $O$. After the sampling of $O$, the $i \in [M]$ that tells which element in the support is picked is provided to the verifier. In particular, the prover does not know $i$, so the proof it sends cannot depend on $i$. Note that the number $[M]$ has to be small, otherwise the verifier cannot read $i$.

\subsection{\texorpdfstring{$\QMA$}{QMA} is contained in \texorpdfstring{$\PSPACE$}{PSPACE}}
In this section, we show that for every distributional oracle $D$, $\QMA^D$ is contained in $\PSPACE^D$. We begin by noting that $\QMA^D \subseteq \PQP^D$ for every distributional oracle $D$, and this proof is the same as Lemma \ref{lem:qmainpqp}, with the quantum oracle replaced by a distributional oracle. We now show that $\PQP^D \subseteq \PSPACE^D$ for every distributional oracle $D$. This is essentially the same argument as \cite{NC16}.

\begin{lemma}\label{lem:pqpinpspacedist}
    For every distributional oracle $D$, $\PQP^D \subseteq \PSPACE^D$. The same holds true for all classical oracles $O$.
\end{lemma}
\begin{proof}
    Let $L \in \PQP^D$ and $V^D$ be the corresponding verifier. We assume that $V^D$ makes all it measurements at the end of the computation, and on an input $x$ of length $n$, it stops in time $p(n)$, where $p: \N \rightarrow \N$ is a polynomial. Then we can think of the action of $V^D$ on the input $x$ as follows:
    \begin{align*}
        V^D(x) = V_{p(n)}O_DV_{p(n)-1}\ldots O_DV_1 \ket{x0^{q(n)}}
    \end{align*}
    where $q(n)$ is the number of ancilla qubits used by $V^D$. Then we define a classical Turing machine $M^D$ which computes the probability that the output of measuring $V^D(x)$ is $y$. This is done by inserting $I = \sum_z \ket{z}\bra{z}$ between every $V_i$ and $O_D$, and then computing the following sum:
    \begin{align*}
        \sum_{z_1, z_2, \ldots z_{2p(n)-2}} \bra{y}V_{p(n)}\ket{z_{2p(n)-2}}\bra{z_{2p(n)-2}}O_D\ket{z_{2p(n)-3}} \ldots \bra{z_2}O_D\ket{z_1} \bra{z_1}V_1\ket{x0^{q(n)}}
    \end{align*}
    This sum can be computed in polynomial space, since $M$ only needs to store a partial sum and $2p(n)-1$ values for computing the products in each sum. Note that to compute these products, $M$ only needs to make classical queries to $O_D$ because $z_i$ are classical strings. Hence $L \in \PSPACE^D$. Note that the proof also works if $D$ is a classical oracle.
\end{proof}

\paragraph{$\TQBF$ completeness fails for distributional oracles.} In the classical oracle model, it is known that if we allow the number of alternating proofs in $\PH$ to grow polynomially with the length of the input, then we get the same class as $\PSPACE$. This follows from Lemma \ref{lem:tqbfpspacecomp}. However, in the case of distributional oracles, the proof from Lemma \ref{lem:tqbfpspacecomp} does not work. In particular, consider $\psi_i$ from Lemma \ref{lem:tqbfpspacecomp} which we defined as follows:
\begin{align*}
        \psi_i(C_1, C_2) := \exists C \forall D_1 \forall D_2 [(D_1 = C_1 \wedge D_2 = C) \vee (D_1 = C \wedge D_2 = C_2)] \Rightarrow \psi_{i-1}(D_1, D_2)
\end{align*}
Recall, that the existentially quantified proof $C$ is the vertex in the configuration graph which has a path of length at most $2^{i-1}$ to both $C_1$ and $C_2$. Therefore, $C$ will be dependent on the oracle, because the oracle responses determine whether there is an edge between two vertices of the configuration graph. If the proof $C$ is only allowed to depend on the distribution that the oracle is sampled from, that may not be sufficient to specify the required vertex in the configuration graph.

\subsection{Oracle Separation between \texorpdfstring{$\QMA$}{QMA} and \texorpdfstring{$\polyQCPH$}{polyQCPH}}
In this section, we consider a distributional oracle separation between $\QMA$ and $\polyQCPH$. A distributional oracle separation between $\QMA$ and $\QCMA$ was first shown by \cite{NN24}. In this section, however, we consider the alternate proof from \cite{LLPY24}. We show that the same proof works to get a distributional oracle separation between $\QMA$ and $\polyQCPH$. We first describe the oracle problem.

\paragraph{\cite{LLPY24} Distributional Oracle Problem.} Fix a code $C \subseteq \Sigma^n$, where $\Sigma = \{0,1\}^{\Theta(n)}$. This is the code $C$ defined in \cite{YZ24}, the properties of which are used in the proof of Lemma \ref{lem:distquerylowerbound} below, but we do not state these properties. Given a function $H: [n] \times \Sigma \rightarrow \{0,1\}$, define a function $f_C^H: C \rightarrow \{0,1\}^n$ as follows:
\begin{align*}
    f_C^H(v_1, v_2, \ldots, v_n) = H(1, v_1) || H(2, v_2)||\ldots ||H(n,v_n)
\end{align*}
We can use this function $f^H_C$ to define a relation $R_H \subseteq \Sigma^n \times \{0,1\}^n$ where $(v,r) \in R_H \Leftrightarrow f_C^H(v) = r$ (this is the same $R_H$ as in the \cite{LLPY24,BK24} classical oracle from Section~\ref{sec:disc}). Then the oracle $O_n^b[H_n, r_n]$ is defined as follows:
\begin{itemize}
    \item (b = 1): $O_n^1[H_n, r_n]$ takes $v \in \Sigma^n$ as input and outputs $1$ if $(v_n, r_n) \in R_H$ and $0$ otherwise.
    \item (b = 0): $O_n^0[H_n, r_n]$ takes $v \in \Sigma^n$ as input and always outputs $0$.
\end{itemize}
The final distribution on the oracles is $\{(O_n^{b_n}[H_n, r_n], r_n)\}_n$. Here the functions $H_n$ parametrize the YES and NO distributions, so the prover will know $H_n$. The support of the distribution is parametrized by the string $r_n$, since this determines the particular values the oracle takes. The strings $r_n \in \{0,1\}^n$ then correspond to $i \in [M]$ and are sampled uniformly at random. The verifier knows $r_n$, but the prover does not.

\begin{lemma}[Lemma 6.1 of \cite{LLPY24}]\label{lem:distquerylowerbound}
    Let $\mathcal{H}_n$ be the set of functions $H: [n] \times \Sigma \rightarrow \{0,1\}$, where $\Sigma = \{0,1\}^{\Theta(n)}$. Then there exists a family of oracles $\{O_n^b[H,r]\}_{n \in \N, b \in \{0,1\}, H \in \mathcal{H}_n, r\in \{0,1\}^n}$ such that for the uniform distribution $\mathcal{D}_n$ over $\mathcal{H}_n \times \{0,1\}^n$:
    \begin{enumerate}
        \item There exist a quantum polynomial-time algorithm $\mathcal{A}$ and quantum witness $\{\ket{z_H}\}_H$ of polynomial size such that
        \begin{align*}
            \Pr_{(H, r) \sim \mathcal{D}_n, b \sim \{0,1\}}[\mathcal{A}^{O_n^b[H,r]}(\ket{z_H}, r) = b] \geq 1 - \mathsf{negl}(n)
        \end{align*}
        \item For any quantum polynomial-time algorithm $\mathcal{B}$, polynomial $s: \N \rightarrow \N$ and for any $s(n)$-bit classical witness $\{z_H\}_H$
        \begin{align*}
            \left|\Pr_{(H, r) \sim \mathcal{D}_n}[\mathcal{B}^{O_n^1[H,r]}(z_H, r) = 1] - \Pr_{(H, r) \sim \mathcal{D}_n}[\mathcal{B}^{O_n^0[H,r]}(z_H, r) = 1]\right| \leq \mathsf{negl}(n)
        \end{align*}
    \end{enumerate}
\end{lemma}

\begin{theorem}
    There is a distributional oracle $O$ relative to which $\QMA^O \not \subset \polyQCPH^O$.
\end{theorem}
\begin{proof}
    Let $L$ be a unary language chosen uniformly at random. For every $n$, set $b_n = 1$ iff $1^n \in L$. Then the oracle $O = \{O_n\}_{n \in \N}$ is sampled as $O_n := (O_n^{b_n}[H_n, r_n], r_n)$ where $(H_n, r_n) \sim \mathcal{D}_n$. Then \cite{LLPY24} show that $L \in \QMA^O$ with probability $1$ over the choice of $\{(b_n, H_n, r_n)\}_{n \in \N}$. Eventually the oracle distribution fixes $H_n$ for every $n$, and so the quantum and classical witnesses for $\QMA^O$ and $\polyQCPH^O$ verifiers are allowed to depend on $H_n$; however, the distribution over $r_n$ remains, so the verifiers are now allowed to depend on $r_n$.\\
    \\
    We now show that $L \notin \polyQCPH^O$ with probability $1$ over the choice of $\{(b_n, H_n, r_n)\}_{n \in \N}$. Fix a verifier $V^O$ which takes $q(n)$ alternatively quantified proofs of length $q(n)$ on inputs of length $n$, where $q:\N \rightarrow \N$ is a polynomial. Let $S_V(n)$ be the event that $V^O$ succeeds on $1^n$, i.e.,
    \begin{align*}
        S_V(n) = [\exists y_1\forall y_2 \ldots Q_{q(n)}y_{q(n)} \Pr[V^{O}(1^n, y_1, \ldots, y_{q(n)}) = 1 | b_n = 1] \geq 2/3] \\ \vee [\forall y_1\exists y_2 \ldots \overline{Q}_{q(n)}y_{q(n)} \Pr[V^{O}(1^n, y_1, \ldots, y_{q(n)}) = 1 | b_n = 0] \leq 1/3]
    \end{align*}
    Then consider an algorithm $\mathcal{B}^{O_n^{b_n}[H_n,r_n]}$ which has $H_i, b_i$ hardcoded in it for $i \neq n$. It randomly samples $r_i$ for $i \neq n$, gets $r_n$ as input and simulates the behaviour of $V^O$. In particular, for every query that $V^O$ makes $O_i$ such that $i \neq n$, $\mathcal{B}$ uses its hardcoded value of $H_i, b_i$ and the sampled value of $r_i$ to generate the query response. If $V^O$ queries a bit of $O_n$, then $\mathcal{B}$ queries its own oracle, and if $V^O$ queries a bit of $r_n$, then $\mathcal{B}$ uses its input value of $r_n$. Then
    \begin{align*}
        \Pr_{H_n, r_n}[\mathcal{B}^{O_n^{b_n}[H_n, r_n]}(y_1, \ldots, y_{q(n)}, r_n) = 1] = \Pr_{\{r_i\}_i, H_n}[V^O(1^n, y_1, \ldots, y_{q(n)}) = 1|\{H_i, b_i\}_{i \neq n}, b_n]
    \end{align*}
    Hence, we have that $\Pr_{\{(b_i, H_i, r_i)\}_i}[S_V(n)] = \frac{1}{2}(\Pr[E_1] + \Pr[E_2])$ where $E_1$ and $E_2$ are the following events (note that the witnesses $y_i$ can depend on $H_n$, but not on $r_n$):
    \begin{align*}
        E_1 &= \exists y_1\forall y_2 \ldots Q_{q(n)}y_{q(n)} \Pr[\mathcal{B}^{O_n^{1}[H_n, r_n]}(y_1, \ldots, y_{q(n)}, r_n) = 1] \geq 2/3 \\
        E_2 &= \forall y_1\exists y_2 \ldots \overline{Q}_{q(n)}y_{q(n)} \Pr[\mathcal{B}^{O_n^{0}[H_n, r_n]}(y_1, \ldots, y_{q(n)}, r_n) = 1] \leq 1/3
    \end{align*}
    We know from Lemma \ref{lem:distquerylowerbound} that for all polynomial sized witnesses $w_{H_n}$,
    \begin{align*}
        \left|\Pr_{H_n, r_n}[\mathcal{B}^{O_n^1[H_n,r_n]}(w_{H_n}, r_n) = 1] - \Pr_{H_n, r_n}[\mathcal{B}^{O_n^0[H_n,r_n]}(w_{H_n}, r_n) = 1]\right| \leq \mathsf{negl}(n)
    \end{align*}
    On applying Markov's inequality,
    \begin{align*}
        \Pr_{H_n, r_n}\left[\left|\Pr[\mathcal{B}^{O_n^1[H_n,r_n]}(w_{H_n}, r_n) = 1] - \Pr[\mathcal{B}^{O_n^0[H_n,r_n]}(w_{H_n}, r_n) = 1]\right| + 1 \geq \frac{5}{4}\right] \leq \frac{4}{5}(1+\mathsf{negl}(n))
    \end{align*}
    Therefore, for at least $1-\frac{4}{5}(1+\mathsf{negl}(n))$ fraction of $(H_n, r_n)$:
    \begin{align*}
        \left|\Pr[\mathcal{B}^{O_n^1[H_n,r_n]}(w_{H_n}, r_n) = 1] - \Pr[\mathcal{B}^{O_n^0[H_n,r_n]}(w_{H_n}, r_n) = 1]\right| < \frac{1}{4}
    \end{align*}
    Therefore, for any such choice of $(H_n, r_n)$, if $E_1$ occurs then $E_2$ does not occur. So $\Pr[E_1 \wedge E_2] \leq \frac{4}{5}(1+\mathsf{negl}(n))$. Thus
    \begin{align*}
        \Pr[S_V(n)] &= \frac{1}{2}(\Pr[E_1] + \Pr[E_2]) \\
        &\leq \frac{1}{2}(\Pr[E_1 \wedge \neg E_2] + \Pr[E_1 \wedge E_2] + \Pr[E_2]) \\
        &\leq \frac{1}{2}(\Pr[\neg E_2] + \frac{4}{5}(1+\mathsf{negl}(n)) + \Pr[E_2]) \\
        &\leq \frac{11}{12}
    \end{align*}
    Then standard arguments can be used (as also done in \cite{LLPY24}) to show that
    \begin{align*}
        \Pr_{\{(b_i, H_i, r_i)\}_i}[\exists V \bigwedge_{n = 1}^\infty S_V(n)] = 0
    \end{align*}
    So $L \notin \polyQCPH^O$ with probability $1$ over the choice of $\{b_n, H_n, r_n\}_n$. By fixing $H_n, b_n$ for every $n$, we get a language $L$ such that $L \in \QMA^O$ and $L \notin \polyQCPH^O$, where $O$ is now the uniformly random distribution over $\{r_n\}_n$.
\end{proof}

\paragraph{Discussion.} The above proof works due to item 2 of Lemma~\ref{lem:distquerylowerbound}. We suspect a statement such as item 2 can only hold for distribution oracles. Indeed, it can be verified for this specific oracle that when the classical proofs $z$ are allowed to depend on $r$ as well as $H$, this statement no longer holds. Since in a standard classical oracle, the proofs would depend on $z$ as well, this proof technique would thus not work for classical oracles.

\section*{Acknowledgements}
We thank Shalev Ben{-}David for discussion on distributional oracles, and William Kretschmer for discussion about the $\QMA$ vs $\QCMA$ candidate in \cite{Zha25} being contained in $\PH$. This work was completed while S. K. was at the Institute for Quantum Computing at the University of Waterloo. A. A. is supported in part by a Cheriton Graduate Scholarship from the School of Computer Science at the University of Waterloo.

\printbibliography

\end{document}

%% file: setup.tex
\usepackage{amsmath,amsthm,amssymb,amsfonts}
\usepackage[american]{babel}
\usepackage{graphicx}
\usepackage[url=false,giveninits=true,maxbibnames=99,style=alphabetic,maxalphanames=5]{biblatex}
\usepackage{mathtools}
\usepackage{enumitem}
\usepackage{csquotes}
\usepackage[noend]{algpseudocode}
\usepackage{algorithm}
\usepackage{xparse}
\usepackage{xspace}
\usepackage{color}
\usepackage{caption}
\usepackage{subcaption}
\usepackage{fullpage}
\usepackage{placeins}
\usepackage{xcolor}
\usepackage{makecell}
\usepackage{qcircuit}
\usepackage{thmtools}
\usepackage{thm-restate}
\usepackage[textsize=scriptsize]{todonotes}
\usepackage{verbatim}
\usepackage{qcircuit}
\usepackage{comment}
\usepackage{mathdots}
\usepackage{hyperref}
\usepackage{cleveref}

\setuptodonotes{color=blue!15}

\mathchardef\mhyphen="2D % Define a "math hyphen"

\newcommand\newmathabbrev[2]{\newcommand{#1}{\ensuremath{#2}\xspace}}

\newcommand\cfont\mathsf
\newmathabbrev\p{\cfont{P}}
\newmathabbrev{\N}{\mathbb N}
\newmathabbrev\NP{\cfont{NP}}
\newmathabbrev\QPH{\cfont{QPH}}
\newmathabbrev\QPHpure{\cfont{pureQPH}}
\newmathabbrev\QCPH{\cfont{QCPH}}
\newmathabbrev\polyQCPH{\cfont{polyQCPH}}
\newmathabbrev\polyPH{\cfont{polyPH}}
\newmathabbrev\QEPH{\cfont{QEPH}}
\newmathabbrev\QMAH{\cfont{QMAH}}
\newmathabbrev\QAC{\cfont{QAC}_0}
\newmathabbrev\AC{\cfont{AC}}
\newmathabbrev\disagr{\mathsf{disagr}}

\newmathabbrev\SIP{\cfont{SIPSER}}

\newmathabbrev\DTIME{\cfont{DTIME}}
\newmathabbrev\tSAT{3\cfont{\mhyphen{}SAT}}
\newmathabbrev\MA{\cfont{MA}}
\newmathabbrev\AM{\cfont{AM}}
\newmathabbrev\NPDAG{\cfont{NP\mhyphen{}DAG}}
\newmathabbrev\QMADAG{\cfont{QMA\mhyphen{}DAG}}
\newmathabbrev\yes{\mathrm{yes}}
\newmathabbrev\no{\mathrm{no}}
\newmathabbrev\US{\cfont{US}}
\newmathabbrev\FP{\cfont{FP}}
\newmathabbrev\PP{\cfont{PP}}
\newmathabbrev\CeP{\cfont{C_=P}}
\newmathabbrev\coCeP{\cfont{coC_=P}}
\newmathabbrev\PH{\cfont{PH}}
\newmathabbrev\SAT{\cfont{SAT}}
\newmathabbrev\SPP{\cfont{SPP}}
\newmathabbrev\GapP{\cfont{GapP}}
\newmathabbrev\BQP{\cfont{BQP}}
\newmathabbrev\ZQEXP{\cfont{ZQEXP}}
\newmathabbrev\QP{\cfont{QP}}
\newmathabbrev\StoqMA{\cfont{StoqMA}}
\newmathabbrev\coNP{\cfont{coNP}}
\newmathabbrev\AzPP{\cfont{A_0PP}}
\newmathabbrev\QMA{\cfont{QMA}}
\newmathabbrev\QXC{\cfont{QXC}}
\newmathabbrev\QMAone{\cfont{QMA}_1}
\newmathabbrev\uQMA{\cfont{uniqueQMA}}
\newmathabbrev\uSAT{\cfont{uniqueSAT}}
\newmathabbrev\cQMA{\cfont{cloneableQMA}}
\newmathabbrev\coQMA{\cfont{coQMA}}
\newmathabbrev\BPP{\cfont{BPP}}
\newmathabbrev\QCMA{\cfont{QCMA}}
\newmathabbrev\pNPlog{\p^{\NP[\log]}}
\newmathabbrev\pNP{\p^{\NP}}
\newmathabbrev\pNPtwo{\p^{\NP[2]}}
\newmathabbrev\pNPone{\p^{\NP[1]}}
\newmathabbrev\pParSAT{\p^{||\SAT}}
\newmathabbrev\pQMApar{\p^{||\QMA}}
\newmathabbrev\pCpar{\p^{||\C}}
\newmathabbrev\pStoqMApar{\p^{||\StoqMA}}
\newmathabbrev\pQMAlog{\p^{\QMA[\log]}}
\newmathabbrev\pClog{\p^{\textup{C}[\log]}}
\newmathabbrev\pC{\p^{\textup{C}}}
\newmathabbrev\QMASPACE{\cfont{QMASPACE}}
\newmathabbrev\pQMAtlog{\p^{\QMA(2)[\log]}}
\newmathabbrev\pStoqMAlog{\p^{\StoqMA[\log]}}
\newmathabbrev\pQMApt{\p^{\Vert\QMA(2)}}
\newmathabbrev\pQMA{\p^{\QMA}}
\newmathabbrev\SharpP{\cfont{\#P}}
\newmathabbrev\pSharP{\p^{\SharpP[1]}}
\newmathabbrev\PromisePP{\cfont{PromisePP}}
\newmathabbrev\lett{\le_\mathrm{tt}}
\newmathabbrev\YES{\mathsf{YES}}
\newmathabbrev\NO{\mathsf{NO}}
\newmathabbrev\PSPACE{\cfont{PSPACE}}
\newmathabbrev\IP{\cfont{IP}}
\newmathabbrev\POLY{\cfont{POLY}}
\newmathabbrev\DAG{\cfont{DAG}}
\newmathabbrev\StoqMADAG{\StoqMA\mhyphen\cfont{DAG}}
\newmathabbrev\CDAG{C\mhyphen\cfont{DAG}}
\newmathabbrev\CDAGf{C\mhyphen\cfont{DAG}_f}
\newmathabbrev\CDAGs{C\mhyphen\cfont{DAG}_s}
\newmathabbrev\CDAGd{C\mhyphen\cfont{DAG}_{d}}
\newmathabbrev\CDAGo{C\mhyphen\cfont{DAG}_1}
\newmathabbrev\LOGS{\cfont{LOGS}}
\newmathabbrev\TAUT{\cfont{TAUTOLOGY}}
\newmathabbrev\SBQP{\cfont{SBQP}}
\newmathabbrev\Fc{F_\coNP}
\newmathabbrev\Fa{F_\AzPP}
\newmathabbrev\GSCON{\cfont{GSCON}}
\newmathabbrev\GSCONexp{\GSCON_\cfont{exp}}
\newmathabbrev\QMAexp{\QMA_\cfont{exp}}
\newmathabbrev\UQMA{\cfont{UQMA}}
\newmathabbrev\R{\mathbb R}
\newmathabbrev\Trees{\cfont{TREES}}
\newmathabbrev\apxsim{\cfont{APX\mhyphen{}SIM}}
\newmathabbrev\AWPP{\cfont{AWPP}}
\newmathabbrev\X{\mathcal{X}}
\newmathabbrev\Y{\mathcal{Y}}

\newmathabbrev\Z{\mathcal{Z}}

\newmathabbrev\ZZ{\mathbb{Z}}
\newmathabbrev\Hprop{H_\mathrm{prop}}
\newmathabbrev\Hin{H_\mathrm{in}}
\newmathabbrev\Hout{H_\mathrm{out}}
\newmathabbrev\Hstab{H_\mathrm{stab}}
\newmathabbrev\Lext{\L_\mathrm{ext}}
\newmathabbrev\BTWNP{\cfont{BTW}(\NP)}
\newmathabbrev\BSN{\cfont{BSN}}
\newmathabbrev\SN{\cfont{SN}}
\newmathabbrev\BD{\cfont{BD}}
\newmathabbrev\HYPERTREE{\cfont{NP\mhyphen{}HYPERTREE}}
\newmathabbrev\Hext{H_\mathrm{ext}}
\newmathabbrev\Hpropt{\tilde{H}_\mathrm{prop}}
\newmathabbrev\Hint{\tilde{H}_\mathrm{in}}
\newmathabbrev\Houtt{\tilde H_\mathrm{out}}
\newmathabbrev\EXP{\cfont{EXP}}
\newmathabbrev\A{\mathcal{A}}
\newmathabbrev\U{\mathcal{U}}

\renewcommand\L{\mathcal{L}}
\newcommand{\E}{\mathbb E}

\newmathabbrev\DAGSSAT{\DAGS(\SAT)}
\newmathabbrev\DAGS{\mathrm{DAGS}}
\newmathabbrev\DAGSNP{\DAGS(\NP)}

\newmathabbrev\AND{\cfont{AND}}

\newmathabbrev\STCONN{{S,T}\cfont{\mhyphen{}CONN}}
\newmathabbrev\CNF{\cfont{CNF}}
\newmathabbrev\NEXP{\cfont{NEXP}}
\newmathabbrev\NPSPACE{\cfont{NPSPACE}}
\newmathabbrev\QCMASPACE{\cfont{QCMASPACE}}
\newmathabbrev\BQPSPACE{\cfont{BQPSPACE}}
\newmathabbrev\PQPSPACE{\cfont{PQPSPACE}}
\newmathabbrev\PQP{\cfont{PQP}}
\newmathabbrev\TQBF{\cfont{TQBF}}
\newmathabbrev{\PCP}{\cfont{PCP}}
\newmathabbrev\BQUPSPACE{\cfont{BQ_UPSPACE}}
\newmathabbrev\QMAt{\QMA(2)}
\newmathabbrev\QMAtSEP{\QMA^{\mathsf{SEP}}(2)}
\newmathabbrev\QMAtexp{\QMAt_{\exp}}
\newmathabbrev\MIP{\cfont{MIP}}
\newmathabbrev\MIPt{\MIP(2)}
\newmathabbrev\BellQMA{\cfont{BellQMA}}
\newmathabbrev\BellQMAt{\BellQMA(2)}
\newmathabbrev\BellQMAtexp{\BellQMAt_{\exp}}

\protected\def\verythinspace{%
  \ifmmode
    \mskip0.5\thinmuskip
  \else
    \ifhmode
      \kern0.08334em
    \fi
  \fi
}

\newcommand{\C}{\mathbb C}

\newcommand{\be}{\begin{equation}}
\newcommand{\ee}{\end{equation}}

\renewcommand{\epsilon}{\varepsilon}

\newcommand{\set}[1]{{\left\{#1\right\}}}    % braces for set notation

\DeclarePairedDelimiter\bra{\langle}{\rvert}
\DeclarePairedDelimiter\ket{\lvert}{\rangle}

\DeclarePairedDelimiterX\braket[2]{\langle}{\rangle}{#1 \delimsize\vert #2}
\DeclarePairedDelimiterX\ketbra[2]{\lvert}{\rvert}{#1 \delimsize\rangle\delimsize\langle #2}

\setlist[itemize]{noitemsep, topsep=0pt}
\setlist[enumerate]{noitemsep, topsep=0pt}

\declaretheorem[numberwithin=section]{theorem}

\declaretheorem[sibling=theorem]{lemma}
\declaretheorem[sibling=theorem]{proposition}

\declaretheorem[sibling=theorem,style=definition]{definition}

\crefname{observation}{observation}{observations}
\Crefname{observation}{Observation}{Observations}

\makeatletter
\newcommand{\subalign}[1]{%
  \vcenter{%
    \Let@ \restore@math@cr \default@tag
    \baselineskip\fontdimen10 \scriptfont\tw@
    \advance\baselineskip\fontdimen12 \scriptfont\tw@
    \lineskip\thr@@\fontdimen8 \scriptfont\thr@@
    \lineskiplimit\lineskip
    \ialign{\hfil$\m@th\scriptstyle##$&$\m@th\scriptstyle{}##$\hfil\crcr
      #1\crcr
    }%
  }%
}
\NewDocumentCommand{\LeftComment}{s m}{%
  \Statex \IfBooleanF{#1}{\hspace*{\ALG@thistlm}}\(\triangleright\) #2}

\def\moverlay{\mathpalette\mov@rlay}
\def\mov@rlay#1#2{\leavevmode\vtop{%
   \baselineskip\z@skip \lineskiplimit-\maxdimen
   \ialign{\hfil$\m@th#1##$\hfil\cr#2\crcr}}}
\newcommand{\charfusion}[3][\mathord]{
    #1{\ifx#1\mathop\vphantom{#2}\fi
        \mathpalette\mov@rlay{#2\cr#3}
      }
    \ifx#1\mathop\expandafter\displaylimits\fi}

\makeatother

\algnewcommand{\LineComment}[1]{\State \(\triangleright\) #1}

\algblockdefx[ON]{Blk}{EndBlk}[1]
  {#1}
  {}

\makeatletter
\ifthenelse{\equal{\ALG@noend}{t}}%
  {\algtext*{EndBlk}}
  {}%
\makeatother

\AtEveryBibitem{%
  \clearlist{language}%
}

\newcommand{\QCSigma}{\mathsf{QC\Sigma}}
\newcommand{\QCPi}{\mathsf{QC\Pi}}
\newcommand{\polyQCSigma}{\mathsf{polyQC\Sigma}}
\newcommand{\polyQCPi}{\mathsf{polyQC\Pi}}

\def\({\left(}
\def\){\right)}
\def\X{\mathcal{X}}
\def\Y{\mathcal{Y}}
\def\Z{\mathcal{Z}}

\def\yes{\text{yes}}
\def\no{\text{no}}

%% file: main.bib
@article{Aar09,
  author       = {Scott Aaronson},
  title        = {On perfect completeness for {QMA}},
  journal      = {Quantum Inf. Comput.},
  volume       = {9},
  number       = {1{\&}2},
  pages        = {81--89},
  year         = {2009},
  url          = {https://doi.org/10.26421/QIC9.1-2-5},
  doi          = {10.26421/QIC9.1-2-5},
  bibsource    = {dblp computer science bibliography, https://dblp.org}
}

@book{AB09,
  author       = {Sanjeev Arora and
                  Boaz Barak},
  title        = {Computational Complexity - {A} Modern Approach},
  publisher    = {Cambridge University Press},
  year         = {2009},
  url          = {http://www.cambridge.org/catalogue/catalogue.asp?isbn=9780521424264},
  isbn         = {978-0-521-42426-4},
  bibsource    = {dblp computer science bibliography, https://dblp.org}
}

@article{AB24,
  author       = {Avantika Agarwal and
                  Shalev Ben{-}David},
  title        = {Oracle Separations for the Quantum-Classical Polynomial Hierarchy},
  journal      = {CoRR},
  volume       = {abs/2410.19062},
  year         = {2024},
  url          = {https://doi.org/10.48550/arXiv.2410.19062},
  doi          = {10.48550/ARXIV.2410.19062},
  eprinttype    = {arXiv},
  eprint       = {2410.19062},
  bibsource    = {dblp computer science bibliography, https://dblp.org}
}

@article{AGKR24,
  author       = {Avantika Agarwal and
                  Sevag Gharibian and
                  Venkata Koppula and
                  Dorian Rudolph},
  title        = {Quantum Polynomial Hierarchies: Karp-Lipton, error reduction, and
                  lower bounds},
  journal      = {CoRR},
  volume       = {abs/2401.01633},
  year         = {2024},
  url          = {https://doi.org/10.48550/arXiv.2401.01633},
  doi          = {10.48550/ARXIV.2401.01633},
  eprinttype    = {arXiv},
  eprint       = {2401.01633},
  bibsource    = {dblp computer science bibliography, https://dblp.org}
}

@article{AHHN24,
  author       = {Anurag Anshu and
                  Jonas Haferkamp and
                  Yeongwoo Hwang and
                  Quynh T. Nguyen},
  title        = {UniqueQMA vs {QMA:} oracle separation and eigenstate thermalization
                  hypothesis},
  journal      = {CoRR},
  volume       = {abs/2410.23811},
  year         = {2024},
  url          = {https://doi.org/10.48550/arXiv.2410.23811},
  doi          = {10.48550/ARXIV.2410.23811},
  eprinttype    = {arXiv},
  eprint       = {2410.23811},
  bibsource    = {dblp computer science bibliography, https://dblp.org}
}

@inproceedings{AK07,
  author       = {Scott Aaronson and
                  Greg Kuperberg},
  title        = {Quantum versus Classical Proofs and Advice},
  booktitle    = {22nd Annual {IEEE} Conference on Computational Complexity {(CCC} 2007),
                  13-16 June 2007, San Diego, California, {USA}},
  pages        = {115--128},
  publisher    = {{IEEE} Computer Society},
  year         = {2007},
  url          = {https://doi.org/10.1109/CCC.2007.27},
  doi          = {10.1109/CCC.2007.27},
  bibsource    = {dblp computer science bibliography, https://dblp.org}
}

@misc{BHNZ25,
      title={Separating QMA from QCMA with a classical oracle}, 
      author={John Bostanci and Jonas Haferkamp and Chinmay Nirkhe and Mark Zhandry},
      year={2025},
      eprint={2511.09551},
      archivePrefix={arXiv},
      primaryClass={quant-ph},
      url={https://arxiv.org/abs/2511.09551}, 
}

@misc{BHV26,
      title={Separating Quantum and Classical Advice with Good Codes}, 
      author={John Bostanci and Andrew Huang and Vinod Vaikuntanathan},
      year={2026},
      eprint={2602.09385},
      archivePrefix={arXiv},
      primaryClass={quant-ph},
      url={https://arxiv.org/abs/2602.09385}, 
}

@inproceedings{BK24,
  author       = {Shalev Ben{-}David and
                  Srijita Kundu},
  editor       = {Karl Bringmann and
                  Martin Grohe and
                  Gabriele Puppis and
                  Ola Svensson},
  title        = {Oracle Separation of {QMA} and {QCMA} with Bounded Adaptivity},
  booktitle    = {51st International Colloquium on Automata, Languages, and Programming,
                  {ICALP} 2024, July 8-12, 2024, Tallinn, Estonia},
  series       = {LIPIcs},
  volume       = {297},
  pages        = {21:1--21:18},
  publisher    = {Schloss Dagstuhl - Leibniz-Zentrum f{\"{u}}r Informatik},
  year         = {2024},
  url          = {https://doi.org/10.4230/LIPIcs.ICALP.2024.21},
  doi          = {10.4230/LIPICS.ICALP.2024.21},
  bibsource    = {dblp computer science bibliography, https://dblp.org}
}

@misc{CKMSTY25,
      title={On the Cryptographic Futility of Non-Collapsing Measurements}, 
      author={Alper Cakan and Dakshita Khurana and Tomoyuki Morimae and Yuki Shirakawa and Kabir Tomer and Takashi Yamakawa},
      year={2025},
      eprint={2510.05055},
      archivePrefix={arXiv},
      primaryClass={quant-ph},
      url={https://arxiv.org/abs/2510.05055}, 
}

@inproceedings{CM24,
  author       = {Andrea Coladangelo and
                  Saachi Mutreja},
  editor       = {Elette Boyle and
                  Mohammad Mahmoody},
  title        = {On Black-Box Separations of Quantum Digital Signatures from Pseudorandom
                  States},
  booktitle    = {Theory of Cryptography - 22nd International Conference, {TCC} 2024,
                  Milan, Italy, December 2-6, 2024, Proceedings, Part {III}},
  series       = {Lecture Notes in Computer Science},
  volume       = {15366},
  pages        = {289--317},
  publisher    = {Springer},
  year         = {2024},
  url          = {https://doi.org/10.1007/978-3-031-78020-2\_10},
  doi          = {10.1007/978-3-031-78020-2\_10},
  bibsource    = {dblp computer science bibliography, https://dblp.org}
}

@article{GMMY24,
  author       = {Eli Goldin and
                  Tomoyuki Morimae and
                  Saachi Mutreja and
                  Takashi Yamakawa},
  title        = {CountCrypt: Quantum Cryptography between {QCMA} and {PP}},
  journal      = {CoRR},
  volume       = {abs/2410.14792},
  year         = {2024},
  url          = {https://doi.org/10.48550/arXiv.2410.14792},
  doi          = {10.48550/ARXIV.2410.14792},
  eprinttype    = {arXiv},
  eprint       = {2410.14792},
  bibsource    = {dblp computer science bibliography, https://dblp.org}
}

@article{GSSSY22,
  author       = {Sevag Gharibian and
                  Miklos Santha and
                  Jamie Sikora and
                  Aarthi Sundaram and
                  Justin Yirka},
  title        = {Quantum generalizations of the polynomial hierarchy with applications
                  to {QMA(2)}},
  journal      = {Comput. Complex.},
  volume       = {31},
  number       = {2},
  pages        = {13},
  year         = {2022},
  url          = {https://doi.org/10.1007/s00037-022-00231-8},
  doi          = {10.1007/S00037-022-00231-8},
  bibsource    = {dblp computer science bibliography, https://dblp.org}
}

@misc{GZ25,
      title={Translating Between the Common Haar Random State Model and the Unitary Model}, 
      author={Eli Goldin and Mark Zhandry},
      year={2025},
      eprint={2503.11634},
      archivePrefix={arXiv},
      primaryClass={quant-ph},
      url={https://arxiv.org/abs/2503.11634}, 
}

@inproceedings{Has86,
  author       = {Johan H{\aa}stad},
  editor       = {Juris Hartmanis},
  title        = {Almost Optimal Lower Bounds for Small Depth Circuits},
  booktitle    = {Proceedings of the 18th Annual {ACM} Symposium on Theory of Computing,
                  May 28-30, 1986, Berkeley, California, {USA}},
  pages        = {6--20},
  publisher    = {{ACM}},
  year         = {1986},
  url          = {https://doi.org/10.1145/12130.12132},
  doi          = {10.1145/12130.12132},
  bibsource    = {dblp computer science bibliography, https://dblp.org}
}

@article{HRST17,
  author       = {Johan H{\aa}stad and
                  Benjamin Rossman and
                  Rocco A. Servedio and
                  Li{-}Yang Tan},
  title        = {An Average-Case Depth Hierarchy Theorem for Boolean Circuits},
  journal      = {J. {ACM}},
  volume       = {64},
  number       = {5},
  pages        = {35:1--35:27},
  year         = {2017},
  url          = {https://doi.org/10.1145/3095799},
  doi          = {10.1145/3095799},
  bibsource    = {dblp computer science bibliography, https://dblp.org}
}

@inproceedings{INNRY22,
  author       = {Sandy Irani and
                  Anand Natarajan and
                  Chinmay Nirkhe and
                  Sujit Rao and
                  Henry Yuen},
  editor       = {Shachar Lovett},
  title        = {Quantum Search-To-Decision Reductions and the State Synthesis Problem},
  booktitle    = {37th Computational Complexity Conference, {CCC} 2022, Philadelphia,
                  PA, USA, July 20-23, 2022},
  series       = {LIPIcs},
  volume       = {234},
  pages        = {5:1--5:19},
  publisher    = {Schloss Dagstuhl - Leibniz-Zentrum f{\"{u}}r Informatik},
  year         = {2022},
  url          = {https://doi.org/10.4230/LIPIcs.CCC.2022.5},
  doi          = {10.4230/LIPICS.CCC.2022.5},
  bibsource    = {dblp computer science bibliography, https://dblp.org}
}

@inproceedings{Kre21,
  author       = {William Kretschmer},
  editor       = {Min{-}Hsiu Hsieh},
  title        = {Quantum Pseudorandomness and Classical Complexity},
  booktitle    = {16th Conference on the Theory of Quantum Computation, Communication
                  and Cryptography, {TQC} 2021, July 5-8, 2021, Virtual Conference},
  series       = {LIPIcs},
  volume       = {197},
  pages        = {2:1--2:20},
  publisher    = {Schloss Dagstuhl - Leibniz-Zentrum f{\"{u}}r Informatik},
  year         = {2021},
  url          = {https://doi.org/10.4230/LIPIcs.TQC.2021.2},
  doi          = {10.4230/LIPICS.TQC.2021.2},
  bibsource    = {dblp computer science bibliography, https://dblp.org}
}

@inproceedings{LLPY24,
  author       = {Xingjian Li and
                  Qipeng Liu and
                  Angelos Pelecanos and
                  Takashi Yamakawa},
  editor       = {Venkatesan Guruswami},
  title        = {Classical vs Quantum Advice and Proofs Under Classically-Accessible
                  Oracle},
  booktitle    = {15th Innovations in Theoretical Computer Science Conference, {ITCS}
                  2024, January 30 to February 2, 2024, Berkeley, CA, {USA}},
  series       = {LIPIcs},
  volume       = {287},
  pages        = {72:1--72:19},
  publisher    = {Schloss Dagstuhl - Leibniz-Zentrum f{\"{u}}r Informatik},
  year         = {2024},
  url          = {https://doi.org/10.4230/LIPIcs.ITCS.2024.72},
  doi          = {10.4230/LIPICS.ITCS.2024.72},
  bibsource    = {dblp computer science bibliography, https://dblp.org}
}

@article{LMY24,
  author       = {Jiahui Liu and
                  Saachi Mutreja and
                  Henry Yuen},
  title        = {{QMA} vs. {QCMA} and Pseudorandomness},
  journal      = {CoRR},
  volume       = {abs/2411.14416},
  year         = {2024},
  url          = {https://doi.org/10.48550/arXiv.2411.14416},
  doi          = {10.48550/ARXIV.2411.14416},
  eprinttype    = {arXiv},
  eprint       = {2411.14416},
  bibsource    = {dblp computer science bibliography, https://dblp.org}
}

@article{MW05,
  author    = {Chris Marriott and
               John Watrous},
  title     = {Quantum Arthur-Merlin games},
  journal   = {Comput. Complex.},
  volume    = {14},
  number    = {2},
  pages     = {122--152},
  year      = {2005}
}

@book{NC16,
  author       = {Michael A. Nielsen and
                  Isaac L. Chuang},
  title        = {Quantum Computation and Quantum Information (10th Anniversary edition)},
  publisher    = {Cambridge University Press},
  year         = {2016},
  url          = {https://www.cambridge.org/de/academic/subjects/physics/quantum-physics-quantum-information-and-quantum-computation/quantum-computation-and-quantum-information-10th-anniversary-edition?format=HB},
  isbn         = {978-1-10-700217-3},
  bibsource    = {dblp computer science bibliography, https://dblp.org}
}

@article{NN24,
  author       = {Anand Natarajan and
                  Chinmay Nirkhe},
  title        = {A distribution testing oracle separation between {QMA} and {QCMA}},
  journal      = {Quantum},
  volume       = {8},
  pages        = {1377},
  year         = {2024},
  url          = {https://doi.org/10.22331/q-2024-06-17-1377},
  doi          = {10.22331/Q-2024-06-17-1377},
  bibsource    = {dblp computer science bibliography, https://dblp.org}
}

@incollection{Wat09,
  author       = {John Watrous},
  editor       = {Robert A. Meyers},
  title        = {Quantum Computational Complexity},
  booktitle    = {Encyclopedia of Complexity and Systems Science},
  pages        = {7174--7201},
  publisher    = {Springer},
  year         = {2009},
  url          = {https://doi.org/10.1007/978-0-387-30440-3\_428},
  doi          = {10.1007/978-0-387-30440-3\_428},
  bibsource    = {dblp computer science bibliography, https://dblp.org}
}

@inproceedings{Zha25,
  author       = {Mark Zhandry},
  editor       = {Raghu Meka},
  title        = {Toward Separating {QMA} from {QCMA} with a Classical Oracle},
  booktitle    = {16th Innovations in Theoretical Computer Science Conference, {ITCS}
                  2025, January 7-10, 2025, Columbia University, New York, NY, {USA}},
  series       = {LIPIcs},
  volume       = {325},
  pages        = {95:1--95:19},
  publisher    = {Schloss Dagstuhl - Leibniz-Zentrum f{\"{u}}r Informatik},
  year         = {2025},
  url          = {https://doi.org/10.4230/LIPIcs.ITCS.2025.95},
  doi          = {10.4230/LIPICS.ITCS.2025.95},
  bibsource    = {dblp computer science bibliography, https://dblp.org}
}

@InProceedings{FK18,
  author =	{Bill Fefferman and Shelby Kimmel},
  title =	{{Quantum vs. Classical Proofs and Subset Verification}},
  booktitle =	{43rd International Symposium on Mathematical Foundations  of Computer Science (MFCS 2018)},
  ISBN =	{978-3-95977-086-6},
  ISSN =	{1868-8969},
  year =	{2018},
  URL =		{http://drops.dagstuhl.de/opus/volltexte/2018/9604},
  URN =		{urn:nbn:de:0030-drops-96040},
  doi =		{10.4230/LIPIcs.MFCS.2018.22}
}

@InProceedings{BFM23,
  author =	{Bassirian, Roozbeh and Fefferman, Bill and Marwaha, Kunal},
  title =	{{On the Power of Nonstandard Quantum Oracles}},
  booktitle =	{18th Conference on the Theory of Quantum Computation, Communication and Cryptography (TQC 2023)},
  ISBN =	{978-3-95977-283-9},
  ISSN =	{1868-8969},
  year =	{2023},
  URL =		{https://drops.dagstuhl.de/entities/document/10.4230/LIPIcs.TQC.2023.11},
  URN =		{urn:nbn:de:0030-drops-183215},
  doi =		{10.4230/LIPIcs.TQC.2023.11}
}

@inproceedings{NZ24,
  author       = {Barak Nehoran and
                  Mark Zhandry},
  title        = {{A Computational Separation Between Quantum No-Cloning and No-Telegraphing}},
  booktitle    = {15th Innovations in Theoretical Computer Science Conference, {ITCS}
                  2024},
  series       = {LIPIcs},
  volume       = {287},
  pages        = {82:1--82:23},
  publisher    = {Schloss Dagstuhl - Leibniz-Zentrum f{\"{u}}r Informatik},
  year         = {2024},
  url          = {https://doi.org/10.4230/LIPIcs.ITCS.2024.82},
  doi          = {10.4230/LIPICS.ITCS.2024.82},
  bibsource    = {dblp computer science bibliography, https://dblp.org}
}

@misc{CKP24,
author = {Rohit Chatterjee and Srijita Kundu and Supartha Podder},
title = {{Are uncloneable proof and advice states strictly necessary?}},
year = {2024},
eprint = {2410.11827},
archivePrefix = {arXiv},
primaryClass = {quant-ph},
url = {https://arxiv.org/abs/2410.11827}
}

@article{Wat03,
  author       = {John Watrous},
  title        = {On the complexity of simulating space-bounded quantum computations},
  journal      = {Comput. Complex.},
  volume       = {12},
  number       = {1-2},
  pages        = {48--84},
  year         = {2003},
  url          = {https://doi.org/10.1007/s00037-003-0177-8},
  doi          = {10.1007/S00037-003-0177-8},
  bibsource    = {dblp computer science bibliography, https://dblp.org}
}

@article{YZ24,
  author       = {Takashi Yamakawa and
                  Mark Zhandry},
  title        = {Verifiable Quantum Advantage without Structure},
  journal      = {J. {ACM}},
  volume       = {71},
  number       = {3},
  pages        = {20},
  year         = {2024},
  url          = {https://doi.org/10.1145/3658665},
  doi          = {10.1145/3658665},
  bibsource    = {dblp computer science bibliography, https://dblp.org}
}

@inproceedings{VV85,
author = {Valiant, L G and Vazirani, V V},
title = {NP is as easy as detecting unique solutions},
year = {1985},
isbn = {0897911512},
publisher = {Association for Computing Machinery},
address = {New York, NY, USA},
url = {https://doi.org/10.1145/22145.22196},
doi = {10.1145/22145.22196},
booktitle = {Proceedings of the Seventeenth Annual ACM Symposium on Theory of Computing},
pages = {458–463},
numpages = {6},
location = {Providence, Rhode Island, USA},
series = {STOC '85}
}

@article{ABOBS22,
  doi = {10.22331/q-2022-03-17-668},
  url = {https://doi.org/10.22331/q-2022-03-17-668},
  title = {The {P}ursuit of {U}niqueness: {E}xtending {V}aliant-{V}azirani {T}heorem to the {P}robabilistic and {Q}uantum {S}ettings},
  author = {Aharonov, Dorit and Ben-Or, Michael and Brand{\~{a}}o, Fernando G.S.L. and Sattath, Or},
  journal = {{Quantum}},
  issn = {2521-327X},
  publisher = {{Verein zur F{\"{o}}rderung des Open Access Publizierens in den Quantenwissenschaften}},
  volume = {6},
  pages = {668},
  month = mar,
  year = {2022}
}

@inproceedings{Sto83,
author = {Stockmeyer, Larry},
title = {The complexity of approximate counting},
year = {1983},
isbn = {0897910990},
publisher = {Association for Computing Machinery},
address = {New York, NY, USA},
url = {https://doi.org/10.1145/800061.808740},
doi = {10.1145/800061.808740},
booktitle = {Proceedings of the Fifteenth Annual ACM Symposium on Theory of Computing},
pages = {118–126},
numpages = {9},
series = {STOC '83}
}

@inproceedings{BCN25,
author = {Bostanci, John and Chen, Boyang and Nehoran, Barak},
title = {Oracle Separation Between Quantum Commitments and Quantum One-Wayness},
year = {2025},
isbn = {978-3-031-91097-5},
publisher = {Springer-Verlag},
address = {Berlin, Heidelberg},
url = {https://doi.org/10.1007/978-3-031-91098-2_1},
doi = {10.1007/978-3-031-91098-2_1},
booktitle = {Advances in Cryptology – EUROCRYPT 2025: 44th Annual International Conference on the Theory and Applications of Cryptographic Techniques, Madrid, Spain, May 4–8, 2025, Proceedings, Part VII},
pages = {3–22},
numpages = {20},
location = {Madrid, Spain}
}

@inproceedings{CCS25,
author = {Chen, Boyang and Coladangelo, Andrea and Sattath, Or},
title = {The Power of a Single Haar Random State: Constructing and Separating Quantum Pseudorandomness},
year = {2025},
isbn = {978-3-031-91097-5},
publisher = {Springer-Verlag},
address = {Berlin, Heidelberg},
url = {https://doi.org/10.1007/978-3-031-91098-2_5},
doi = {10.1007/978-3-031-91098-2_5},
booktitle = {Advances in Cryptology – EUROCRYPT 2025: 44th Annual International Conference on the Theory and Applications of Cryptographic Techniques, Madrid, Spain, May 4–8, 2025, Proceedings, Part VII},
pages = {108–137},
numpages = {30},
location = {Madrid, Spain}
}

@inproceedings{BMMMY25,
author = {Behera, Amit and Malavolta, Giulio and Morimae, Tomoyuki and Mour, Tamer and Yamakawa, Takashi},
title = {A New World in the Depths of Microcrypt: Separating OWSGs and Quantum Money from QEFID},
year = {2025},
isbn = {978-3-031-91097-5},
publisher = {Springer-Verlag},
address = {Berlin, Heidelberg},
url = {https://doi.org/10.1007/978-3-031-91098-2_2},
doi = {10.1007/978-3-031-91098-2_2},
booktitle = {Advances in Cryptology – EUROCRYPT 2025: 44th Annual International Conference on the Theory and Applications of Cryptographic Techniques, Madrid, Spain, May 4–8, 2025, Proceedings, Part VII},
pages = {23–52},
numpages = {30},
location = {Madrid, Spain}
}

@inproceedings{AGL24a,
author = {Ananth, Prabhanjan and Gulati, Aditya and Lin, Yao-Ting},
title = {Cryptography in the Common Haar State Model: Feasibility Results and Separations},
year = {2024},
isbn = {978-3-031-78016-5},
publisher = {Springer-Verlag},
address = {Berlin, Heidelberg},
url = {https://doi.org/10.1007/978-3-031-78017-2_4},
doi = {10.1007/978-3-031-78017-2_4},
booktitle = {Theory of Cryptography: 22nd International Conference, TCC 2024, Milan, Italy, December 2–6, 2024, Proceedings, Part II},
pages = {94–125},
numpages = {32},
location = {Milan, Italy}
}

@misc{AGL24b,
      title={A Note on the Common Haar State Model}, 
      author={Prabhanjan Ananth and Aditya Gulati and Yao-Ting Lin},
      year={2024},
      eprint={2404.05227},
      archivePrefix={arXiv},
      primaryClass={quant-ph},
      url={https://arxiv.org/abs/2404.05227}, 
}

@misc{CGG24,
      title={On Central Primitives for Quantum Cryptography with Classical Communication}, 
      author={Kai-Min Chung and Eli Goldin and Matthew Gray},
      year={2024},
      eprint={2402.17715},
      archivePrefix={arXiv},
      primaryClass={cs.CR},
      url={https://arxiv.org/abs/2402.17715}, 
}

@article{BGS75,
author = {Baker, Theodore and Gill, John and Solovay, Robert},
title = {Relativizations of the P =? NP Question},
journal = {SIAM Journal on Computing},
volume = {4},
number = {4},
pages = {431-442},
year = {1975},
doi = {10.1137/0204037},
URL = {https://doi.org/10.1137/0204037},
eprint = {https://doi.org/10.1137/0204037}
}

@article{BBBV97,
   title={Strengths and Weaknesses of Quantum Computing},
   volume={26},
   ISSN={1095-7111},
   url={http://dx.doi.org/10.1137/S0097539796300933},
   DOI={10.1137/s0097539796300933},
   number={5},
   journal={SIAM Journal on Computing},
   publisher={Society for Industrial & Applied Mathematics (SIAM)},
   author={Bennett, Charles H. and Bernstein, Ethan and Brassard, Gilles and Vazirani, Umesh},
   year={1997},
   month=oct, pages={1510–1523}
}

@inproceedings{RT19,
  title = {Oracle {{Separation}} of {{BQP}} and {{PH}}},
  booktitle = {Proceedings of the 51st {{Annual ACM SIGACT Symposium}} on {{Theory}} of {{Computing}}},
  author = {Raz, Ran and Tal, Avishay},
  year = {2019},
  series = {{{STOC}} 2019},
  pages = {13--23},
  publisher = {{Association for Computing Machinery}},
  address = {{New York, NY, USA}},
  doi = {10.1145/3313276.3316315},
  isbn = {978-1-4503-6705-9}
}
